\documentclass[11pt]{article}

\usepackage{natbib,stfloats}
\usepackage{mathrsfs,amsmath,upgreek,amssymb,amsthm}
\usepackage{multirow,a4wide}

\usepackage[english]{babel}

\usepackage{sgamevar}
\usepackage{xspace}

\newcommand{\ang}{\mathfrak{a}}
\newcommand{\dae}{\mathfrak{d}}
\newcommand{\ecomodel}{\mathfrak{M}}
\newcommand{\pertur}{\mathcal{S}}
\newcommand{\angdae}{$\mathfrak{a}/\mathfrak{d}$\xspace}
\newcommand{\agent}{\mathfrak{p}}
\newcommand{\stress}{\textsf{strength}}

\newcommand{\expect}{\mathbb{E}}

\newcommand{\PNE}{\hbox{\textsc{PNE}}\xspace}
\newcommand{\NE}{\hbox{\textsc{NE}}\xspace}
\newcommand{\DSE}{\hbox{\textsc{DSE}}\xspace}

\newcommand{\remove}[1]{}

\def\newblock{}

\newcommand{\be}{\begin{eqnarray}}
\newcommand{\ee}{\end{eqnarray}}

\renewcommand{\arraystretch}{1.2}

\newtheorem{lemma}{Lemma}
\newtheorem{theorem}[lemma]{Theorem}

\newtheorem{property}[lemma]{Property}

\newtheorem{definition}{Definition}

\newtheorem{ex}{Example}[section]

\title{Uncertainty  Analysis  of Simple  Macroeconomic Models\\ Using Angel-Daemon Games}
\author{Joaquim  Gabarro \and Maria Serna}

\date{}
\begin{document}%

\maketitle

\vspace{-1cm}
\begin{center}
{Computer Science  Dept. \\
ALBCOM Research Group,\\
Universitat Polit\`ecnica de Catalunya, BarcelonaTech,  
Barcelona, Spain\\
E-mail: gabarro,mjserna@cs.upc.edu
}
\end{center}

\begin{abstract}
We propose the use of   an angel-daemon framework to perform an  uncertainty analysis of  short-term macroeconomic models  with exogenous components.  
An uncertainty profile $\mathcal U$ is  a short and macroscopic description of a potentially perturbed situation. The angel-daemon framework  uses  $\mathcal U$  to define a strategic game
where two agents, the angel and the daemon, act selfishly having different  goals. The  Nash equilibria of those games provide the stable strategies in  perturbed  situations, giving a natural estimation of uncertainty.
 In this initial work we apply the framework in order to get an 
 uncertainty  analysis of linear versions of the IS-LM and the IS-MP models. In those models, by considering uncertainty profiles, we can capture different economical situations.
Some of them  can be described in terms of macroeconomic policy coordination. In other cases we just  analyse the results of the system under some possible perturbation level.  
Besides providing examples of application  we  analyse the structure of the Nash equilibria in some particular cases of interest. 
\end{abstract}




\section{Introduction}
\label{Easy IS-LM Model}

The distinction between \emph{risk} and \emph{uncertainty} has become increasingly  important  since 
\citep{Knight} discussed it as we have imperfect knowledge of future events in our ever-changing world. 
Informally, \emph{risk} can be measured by probabilities. In contrast, \emph{uncertainty} refers to something where we cannot even gather   the information required to figure out probabilities.
However, in practice there is no difference between 
risk and uncertainty in empirical analysis on the economy and 
financial markets. Both  are   measured by 
historical standard deviation of the variable of interest \citep{HullRisk,ArgimiroBook}. 
This paper proposes an alternative to disentangle these 
seemingly indistinguishable concepts applying ideas from game theory and computer science.

The study of web applications is a field where uncertainty becomes  unavoidable.
The angel-daemon  framework \citep{GSS2014-CJ} provides a way to obtain  numerical estimates of uncertainty  in the execution of a Web service. In such a setting, the uncertainty  is  captured by an uncertainty profile describing a  stressed environment for the execution of the Web application. Uncertainty profiles provide a description of the perceived  uncertain behaviour with respect to possible failing services or execution delays.   
That is, some sites can potentially misbehave but we are uncertain about the specific sites that will do so. 
The model attempts to  balance positive and negative aspects. Considering  only positive aspects (minimizing damage) is usually  too optimistic. In the opposite side, being pessimistic (maximizing damage) is also  not realistic. Reality often evolves in between optimism and pessimism. To model this situation, the framework considers two agents : \emph{the angel}  ($\ang$),  dealing  with the optimistic side; and  the \emph{daemon} ($\dae$), dealing with the pessimistic side. These agents  act strategically in an associated angel-daemon zero-sum game.
In this context, uncertain situations are identified with the  Nash equilibria of the angel-daemon game and they are assessed by the value of the game.  
It is important to emphasize that the results in  \citep{GSS2014-CJ}  are useful to analyse  uncertain stable (or timeless) environments. Thus, the framework analyses uncertainty in the short-term and it  is not useful for a long-term analysis.

In this paper, we present an angel-daemon ($\ang/\dae$) framework to model 
uncertainty  in short-term macroeconomic models.
In the $\ang/\dae$ approach the actions undertaken by $\ang$ and $\dae$ usually go into different  directions and  can affect 
the result  of underlying systems in unexpected ways. 
In some cases, $\ang$ or $\dae$ may be identified with policy makers or institutions. In other cases,  they may describe a situation created by many interacting agents.
We have to be careful about some facts. 
%
First, it is usually mostly difficult to  picture a policy-maker or an institution acting deliberately as $\dae$. 
Second, the distinction $\ang/\dae$ at first sight seems unsuitable when it comes to market
participants, as there is not an obvious pessimistic party opposed 
to an openly optimistic party. Despite the judging names, it is not our intention to associate any moral connotation to them. However,  there are actors that might be thought of as performing the role of $\ang$  or the  $\dae$, even with moral issues \citep{AkeShi,PhishingPools}. 
 We clarify our approach using  an example considering  the role of  the Fed and Wall Street  in the Great Recession in terms of $\ang/\dae$.
%
The Fed, as policy-maker, who tried to improve the economic behaviour,  is identified with $\ang$. Wall Street (representing other economic agents) used this policy strategically for their own interest and we identify  them with  $\dae$.
However,  Wall Street ($\dae$) interacted with Fed ($\ang$)  using a policy in an  unexpected way \citep{GreenspanNeverSaw}.
%
It is reasonable to assume that  the Great Recession appears as a consequence of the a priori broadly unexpected and strategic interplay between $\ang$ and $\dae$ \citep{LetterQueen}.
%

Our approach may provide another  interpretation of the Keynesian animal spirits  \citep{KeyEmp}.  As Keynes pointed out,    
our decisions to do something are mostly the result of  a ''spontaneous urge to act" (animal spirits).
Thus, in  these models,  
animal spirits ($\ang$ and $\dae$) can be thought as players in a strategic game. 
The characters representing $\ang$ and $\dae$ are summarized into
the utilities (or dis-utilities) $u_{\ang}$ and $u_\dae$.  Giving different values to $u_\ang$ and $u_\dae$
it is possible to shape different situations.
For instance, if we set $u_{\ang}=Y$ and $u_{\dae}= T$, the angel tries to maximize the income (in many cases considered as a good issue) and the daemon tries to maximize taxes (considered a bad issue in some schools).  By setting  $u_\ang =-T$ and $u_\dae = r$, the angel tries to maximize $-T$ (minimizing $T$) and $\dae$ tries to maximize the interest rate.  

In order to  link the $\ang/\dae$ framework with short-term macroeconomic models, we identify the parts of the macroeconomic system to be perturbed.
We  assume that $\ang$ and $\dae$ can exert some power in the value of  some of the exogenous components of the model.  We model the power of  action by potential perturbations that $\ang$ and $\dae$ might apply  to modify the components'  estimation.
The perturbation values are real numbers, so they can be either positive or negative. 
We also assume that both agents have limits in the influence they can exert.  This gives rise to an adequate redefinition of uncertainty profiles describing the component's variability in an perturbed situation.  As in \citep{GSS2014-CJ}, once the uncertainty profile is defined, we analyse the stable situations of the corresponding  \angdae-game.  
The obtained \angdae games have a richer structure than  the initial application as  the potential strategic situation cannot always be model by a zero-sum game. Nevertheless,  we can define some interesting valuations that can be analysed through zero-sum \angdae games.

To test the applicability  of the $\ang/\dae$ framework,  we start with linear approximation of extensively studied models. In particular, we develop $\ang/\dae$ analysis of  the InvestmentSavings-LiquidityMoney (IS-LM) introduced by \citep{Hicks,Hicks80-81} and  the InvestmentSavings-MonetaryPolicy  (IS-MP) developed by \citep{DRomer2000}. 
We have selected those models as they provide different interpretations of the monetary policy, differentiating periods  without inflation from those  with inflation.  
Besides  defining the framework we model
uncertainty of the fiscal policy in the IS-LM model and of the external shocks in the IS-MP model. 
We either obtain the Nash equilibria or we analyse the properties of Nash equilibria.
In particular,  we show that if we are  uncertain of the fiscal policy, there is always a  a dominant strategy equilibria.

The paper is structured as follows. 
In Section  \ref{sec:bmemod} we briefly describe the IS-LM and the IS-MP models. Section  \ref{sec:laEG} is devoted to  the formulation of the linear approximations to the IS-LM and  the IS-MP models describing their exogenous components. In Section \ref{Stress}, we provide a model for the possible perturbations  of the set of exogenous components,  the so called \emph{perturbation strength model}.  Section \ref{Uncertainty} introduces uncertainty profiles and the associated \angdae games tailored  to the linear  IS-LM and IS-MP models and analyses the Nash equilibria of  some cases. Section \ref{section-IS-LM} studies the IS-LM model when we are uncertain of the fiscal policy and Section \ref{section-IS-MP} studies the IS-MP model when we are uncertain of the external shocks. 
Finally, in Section \ref{Conclusions}, we raise some concluding remarks and discuss some future research.


\section{Basic Macroeconomic Models}
\label{sec:bmemod}
In order to make this paper self contained and accessible to a wide audience, we describe briefly 
the two basic Macroeconomic models considered in this paper.
 The IS-LM which  is useful to describe periods with no inflation  and the IS-MP in which the inflation is taken explicitly into account by the  monetary policy fixed by the Central Bank.
Recall that, according to Fisher's equation, the relation between the real interest rate $r$, the nominal interest rate $i$ and the inflation $\pi$ is 
$r=i-\expect(\pi)$ where $\expect(\pi)$ is the inflationary expectations. When there is no inflation, the real and nominal interest rates coincide.  

\subsection{IS-LM model}
The 
InvestmentSavings-LiquidityMoney model (IS-LM) \citep{Hicks} provides  a way to express, in equilibrium, the national income and the interest rate as a a function of several exogenous components. 
The IS-LM model describes approximately the monetary market when the gold standard was the norm because the gold can be reasonably represented by the 
money supply $M$.  Despite its simplicity,  this model continues making useful predictions \citep{IS-LMentary}.  
As we have mention before,  when there is no inflation, the nominal and the real interest rate coincide, so only $r$ appears in the equations.  This situation of no inflation might happen for long periods,
the  important fact to remember  is that inflation is mostly  a twentieth-century
phenomenon. Up to World War I, inflation was zero or close to it \citep{PikettyCapital}. 
So, for some time, national income $Y$ and interest rate $r$  were the two main macroscopic variables. 
The IS-LM is described by two equations.

\begin{itemize}
\item
The IS line $Y= C(Y-T)+ I(r) + G$ represents a continuum of equilibria in 
the goods market.   
On it, $Y$ is the national income,  $r$ is the interest rate.  The remaining components are the sum of the annual rates of spending by:  the consumers  (as a function of the disposable income) $C(Y-T)$; the investors (as a function of the interest rate) $I(r)$ and  the government $G$. 
\item
The  LM line $M/P=L(r,Y)$ is interpreted as continuum of equilibrium in the money market.
The money supply is $M/P$, where  $M$ is the money  and $P$ is the price level.  The 
liquidity preference $L(r,Y)$ is a function of the national income and the interest rate.  
\end{itemize}
An equilibrium point $(Y,r)$ is a solution of the system of equations. These equilibria  correspond to the  points where both markets are at mutual equilibrium. 
%


\subsection{IS-MP Model}
In the thirties, the world was in transition from the gold standard  and the monetary policy of the  central banks changed  to deal with inflation.
In calm periods, central banks  ensure that the money supply grows as
economic activity in order  to guarantee a low inflation rate of 1 or 2 percent a year. The Central Bank 
creates new money by lending to banks for very short periods 
\citep{PikettyCapital}.
In Europe, the primary objective of the monetary policy of the European Central Bank is to maintain price stability as a way to contribute to economic growth and job creation. In the pursuit of price stability, the European Central Bank aims at maintaining inflation rates below, but close to, 2\% over the medium term \citep{ECB-MonetaryPolicy}.
The InvestementSaving-MonetaryPolicy model (IS-MP) \citep{DRomer2000} considers this new reality about monetary policy and inflation. The MP equation deals with the interest rate as a function of the current inflation, the Central Bank expected inflation an  the income gap \citep{Taylor1993}. 

We start with a short description of the {\em dynamic aggregate demand/aggregate supply model} (dynamic AD/AS model)   given in \citep{Mankiw}. The model is dynamic because some variables depend on their lagged (past period) values. In the model $t$ denotes the current period  (usually one year) and is given by the following equations.
\begin{itemize}
\item
The \emph{supply for goods and services}  is given by 
$Y_t = \overline{Y}_t-\alpha(r_t-\rho)+\epsilon_t$.
In this equation, 
the total output for goods and services is $Y_t$  and  the economy's natural output is  $\overline{Y}_t$.
The parameter $\alpha>0$ measures the sensitivity of the demand in front of the real interest rate $r_t$ and $\rho$
is the natural rate of interest.
The parameter $\epsilon_t$ represents the random demand shock. 
\item
The \emph{real interest rate} $r_t$ is given by a simplified version of Fisher's equation with no expectations: $r_t= i_t-\pi_t$.  Where  
$i_t$ is the nominal interest rate and $\pi_t$ is the inflation rate. 
\item
\emph{Philips curve}. The inflation $\pi_t$ at period $t$ is described by a version of the Phillips curve as a function of the past inflation $\pi_{t-1}$ and with no expectations: 
$\pi_t=\pi_{t-1}+\phi(Y_t-\overline{Y}_t)+v_t$. The parameter $\phi>0$ measures the responsiveness of the inflation to output fluctuations and $v_t$ is the random supply shock.
\item
{\sl Monetary policy rule}. It is based on Taylor's Rule \citep{Taylor1993}. The nominal interest rate $i_t$ is given by 
$i_t= \pi_t+\rho + \theta_\pi(\pi_t-\pi^*_t)+ \theta_Y(Y_t-\overline{Y}_t)$. In this equation, $\pi^*_t$ is the central bank's target inflation rate and $\theta_\pi>0$, $\theta_Y>0$ measure responsiveness.  
\end{itemize}
We avoid temporal dependencies and we consider a version of the IS-MP model assuming that past inflation coincides with the Central Bank target inflation. 
As before an equilibrium point  $(Y,\pi)$ is a solution of the system of equations.  
%

\section{Linear Approximations and Exogenous Components}
\label{sec:laEG}
We introduce here the linear approximations of the IS-LM and the IS-MP models. For both simplified models we make explicit their exogenous  components. Let us first fix some notation. We use $\ecomodel$ to denote a  linear approximation of a model, informally  $\ecomodel \in\{\hbox{IS-LM}, \hbox{IS-MP}\}$. For a model $\ecomodel$,
$\mathcal P_\ecomodel$   denotes the set of exogenous parameters,  $\mathcal V_\ecomodel$ denotes the the set of  exogenous variables and  $\mathcal E_\ecomodel=\mathcal P_\ecomodel\cup \mathcal V_\ecomodel$ is the set  of \emph{exogenous components}. We use set notation like $b\in \mathcal P_\ecomodel$ or $T\in \mathcal V_\ecomodel$.
When $\ecomodel$ is clear from the context we use  $\mathcal E$, $\mathcal P$ and $\mathcal V$.


\begin{figure}[t]
\begin{center}
{\renewcommand{\arraystretch}{1.2}
\begin{tabular}{||l|c||}
\hline
{\sf Variables} &$\mathcal V$\\
\hline
Taxes & $0<T$\\
Exogenous  government spending & $0<G$\\
Money Supply & $0<M$\\
Price index &$0<P$\\
\hline
{\sf Parameters} &$\mathcal P$\\
\hline
Autonomous  consumption & $0<a$ \\
Marginal propensity to consume & $0< b< 1$\\
Exogenous investment & $0<c$\\
Interest sensitivity &  $0<d$\\
Income sensitivity for real money& $0<e$\\
Interest sensitivity for real money&  $0<f$\\
\hline
\end{tabular}
}
\end{center}
\caption{The exogenous components in the linear approximation to the IS-LM model. \label{exogenous}}
\end{figure}

\subsection{IS-LM  Model}
\label{Model}
We consider the linear approximation  of the IS-LM given in \citep{BalBraTur}. 
\begin{definition}
The IS-LM model is described by the following equations.
\[C(Y-T)= a+ b(Y-T),\; I(r) =c-d\, r,\;  L(r,Y)=eY-f\, r.\]  
The set \emph{exogenous components}  $\mathcal E=\mathcal V\cup\mathcal P$  is given by 
$\mathcal V= \{a, b, c,d, e, f\}$ and  $\mathcal P= \{T,G, M, P\}$ (see Figure \ref{exogenous}). 
\end{definition}
Let us express the  \emph{endogenous variables}  $\{Y, r\}$  in equilibrium as a function of $\mathcal E$.
The equilibrium condition  $Y=Y(r)$ (IS line) gives the equation 
$Y=a+ b(Y-T)+c-d\, r+G$. The condition  $r=r(Y )$ (MP line)  gives  
$M/P= eY-f\, r$. Thus, we get
\[Y= \frac{1}{(1-b)}(a+c + G - b\,T-d\, r)\quad 
\hbox{and}\quad  r= \frac{1}{f}(e\,Y- \frac{M}{P}).
\]
Using matrix notation, the  linear system  describing  $Y(r)$ and $r(Y)$  can be  written as
\[\left(
\begin{array}{cc}
1-b & d\\
P\, e &\hbox{\;\;}-P\, f
\end{array}
\right)
\left(
\begin{array}{c}
Y\\
r
\end{array}
\right)
=
\left(
\begin{array}{c}
a+c+G-b\, T\\
M
\end{array}
\right).
\]
Solving the system, we get the following expression for the equilibrium point $(Y, r)$. 
\[
\left(
\begin{array}{c}
Y\\
r
\end{array}
\right)
=
\frac{1}{(1-b)f+d\,e}\left(
\begin{array}{cc}
f & d/P\\
e &\hbox{\;\;}-(1-b)/P
\end{array}
\right)
\left(
\begin{array}{c}
a+c+G-b\, T\\
M
\end{array}
\right).
\]
Defining $g=(1-b)f+d\, e$, we get the following expressions.
\begin{eqnarray*}
&&Y=\frac{f}{g}(a+c+G-bT) + \frac{d}{g}\frac{M}{P} \hbox{\; and\;}
r=\frac{e}{g}(a+c+G-bT)-\frac{(1-b)}{g}\frac{M}{P}.
\end{eqnarray*}
As $(Y,r)$ depends on the valuation on $\mathcal E$, when needed we write
$(Y(\mathcal E), r(\mathcal E))$.  In order to simplify notation we often use  $\mathcal E$ to refer also to a valuation of the exogenous components of the model. 

\begin{ex}
\label{EasyExample}
\em
\label{Mankiw}
Consider the following valuation of $\mathcal E$:
 \[
{\renewcommand{\arraystretch}{1.2}
\begin{tabular}{|| c | c  |c | c | c | c | c | c | c | c ||}
\hline
$a$ & $b$ & $c$ & $d$ & $e$ & $f$ & $T$ & $G$ & $M$ & $P$\\
\hline
$200$ & $3/4$ & $200$ & $25$ & $1$ & $100$ & $100$ & $100$ & $1000$ & $2$\\
\hline
\end{tabular}
}
\]
The equilibrium point is described by  the linear system $Y= 1700-100r$,
$r=Y/100 - 5$.  Solving the system we get  
$Y= 1100$, $r=6$.
\hfill$\Box$
\end{ex}
%


\begin{figure}[t]
\begin{center}
{\renewcommand{\arraystretch}{1.2}
\begin{tabular}{||l|c||}
\hline
{\sf Variables} &$\mathcal V$\\
\hline
Central bank's target inflation & $\pi^*$\\
Natural level of output & $\overline Y$\\
Shock to $Y$ & $\epsilon$\\
Shock to $\pi$ & $v$\\
\hline
{\sf Parameters} & $\mathcal P$ \\
\hline
$Y$ sensitivity to $r$ & $0<\alpha$ \\
Natural interest rate & $0<\rho$\\
$\pi$ sensitivity to $Y$ in Philips line  &  $0<\phi$\\
$i$ sensitivity to inflation in MP & $0<\theta_\pi$\\
$i$ sensitivity to $Y$ in MP &  $0<\theta_Y$\\
\hline
\end{tabular}
}
\end{center}
\caption{
The exogenous components in the linear approximation to the  IS-MP model. \label{exogenous-IS-MP}}
\end{figure}

\subsection{IS-MP Model}
\label{sec-IS-MP}

We present now the linear approximation to the IS-MP model  from the dynamic AS/AD model as given in \citep{Mankiw}. Recall that  the equations on $Y_t$, $r_t$ and $i_t$ give us the following dynamic aggregate demand line
\[Y_t = \overline{Y}_t-\hat \alpha (\pi_t-\pi^*_t) + \hat \beta \epsilon_t
\hbox{\; where\;} 
\hat \alpha = \frac{\alpha\theta_\pi}{1+\alpha\theta_Y}
\hbox{\; and\;} 
\hat \beta = \frac{1}{1+\alpha\theta_Y}\]
and the aggregate {dynamic aggregate supply line} is given by
\[\pi_t=\pi_{t-1}+\phi(Y_t-\overline{Y}_t)+v_t.\]

We consider a  simplified linear approximation to the IS-MP in which we avoid  
dependences of inflation on their lagged values.  We consider only the case where the past inflation coincides with the Central Bank target inflation, i.e. $\pi_{t-1} = \pi^*_t$.  Furthermore, assuming that period $t$ is known, we consider the equations for a fixed period $t$ so that we can drop  the sub-index.

%
\begin{definition}
The  IS-MP model  is described by the following equations.
\begin{equation*}
Y= \overline{Y}-\hat \alpha (\pi-\pi^*)+\hat \beta \epsilon 
\hbox{\; and\; }
\pi=\pi^*+\phi(Y-\overline{Y})+v
\end{equation*}
The  set  of {\em exogenous components}  $\mathcal E=\mathcal V \cup \mathcal P$ in the  IS-MP model  is given by
$\mathcal V =\{\pi^*, \overline{Y},\epsilon, v\}$ and  $\mathcal P = \{\alpha, \rho, \phi,\theta_\pi,\theta_Y\}$ (see Figure~\ref{exogenous-IS-MP}).
\end{definition}
From the description of the equations of the IS-LM model we can write the equations of an equilibrium  for the endogenous variables  $\{Y, \pi\}$.  Solving the system, as a function of $\mathcal E$, the equilibrium point is given by 
\begin{equation*}
Y = \overline{Y} + \hat \gamma \epsilon -\hat \delta v
\hbox{\; and\; }
\pi = \pi^* + \hat \rho \epsilon + \hat \mu v
\end{equation*}
where
\begin{eqnarray*} 
&&
\hat \gamma 
= \frac{\hat \beta}{1+\hat\alpha\phi}=
\frac{1}{1+\alpha(\theta_Y+\phi\theta_\pi)}
\hbox{\; and\;} 
\hat\delta 
= \frac{\hat \alpha}{1+\hat \alpha \phi} 
=\frac{\alpha\theta_\pi}{1+\alpha(\theta_Y+\phi\theta_\pi)}\\
&&\hat \rho = \frac{\phi\hat\beta}{1+\hat\alpha\phi} = \frac{\phi}{1+\alpha(\theta_Y+\phi\theta_\pi)}
\hbox{\; and\;} 
\hat \mu = \frac{1}{1+\hat\alpha\phi}= \frac{1+\alpha\theta_Y}{1+\alpha(\theta_Y+\phi\theta_\pi)}.
\end{eqnarray*}
The preceding equations can be  written  in matrix form as
\[
\left( 
\begin{array}{c}
Y\\
\pi
\end{array}
\right)=
\left( 
\begin{array}{rr}
\hat\gamma &\hbox{\;\;}-\hat\delta\\
\hat \rho & \hat\mu
\end{array}
\right)
\left( 
\begin{array}{c}
\epsilon\\
v
\end{array}
\right)+
\left( 
\begin{array}{c}
\overline{Y}\\
\pi^*
\end{array}
\right)
\]
or,  alternatively as 
\[
\left( 
\begin{array}{c}
Y\\
\pi
\end{array}
\right)=
\frac{1}{1+\alpha(\theta_Y+\phi\theta_\pi)}
\left( 
\begin{array}{rr}
1 &-\alpha\theta_\pi\\
\phi & \hbox{\;\;}1+\alpha\theta_Y
\end{array}
\right)
\left( 
\begin{array}{c}
\epsilon\\
v
\end{array}
\right)+
\left( 
\begin{array}{c}
\overline{Y}\\
\pi^*
\end{array}
\right).
\]
The equilibrium point is $(Y,\pi)$ is a function of the valuation $\mathcal E$. As before,  when needed we write $(Y(\mathcal E), \pi(\mathcal E))$.
\begin{ex}
\label{first-IS-MP}
Consider the following $\mathcal E$ for the IS-MP model:
\[
{\renewcommand{\arraystretch}{1.2}
\begin{tabular}{|| c | c  |c | c | c | c | c | c | c  ||}
\hline
$\alpha$ & $\rho$ & $\phi$ & $\theta_\pi$ & $\theta_Y$ & $\pi^*$ & \rule{0pt}{11pt}$\overline{Y}$ & $\epsilon$  & $v$\\
\hline
$1$ & $2$ & $1/4$ & $1/2$ & $1/2$ & $2$ & $100$ & $1$ & $1/2$ \\
\hline
\end{tabular}
}
\]
We can compute directly 
\[
\left( 
\begin{array}{c}
Y\\
\pi
\end{array}
\right)=
\frac{1}{13}
\left( 
\begin{array}{rr}
8 &-4\\
2 & \hbox{\;\;}12
\end{array}
\right)
\left( 
\begin{array}{c}
1\\
1/2
\end{array}
\right)+
\left( 
\begin{array}{c}
100\\
2
\end{array}
\right)=
\frac{1}{13}
\left( 
\begin{array}{c}
1306\\
34
\end{array}
\right)
\]
Therefore $Y= 1306/13$ and $\pi=34/13$.
\hfill$\Box$
\end{ex}


\section{Perturbation Strength Model}
\label{Stress}
Given $\ecomodel$ and a valuation $\mathcal E$ of its exogenous components, the computation of some positive aspects might have been underestimated or some  negative aspects overestimated. 
Therefore it makes sense to study $\ecomodel$ under slight (or severe) perturbations. This is achieved studying $\ecomodel$ when the valuation $\mathcal E$ is modified strategically.  The first component of our proposal is a description of the perceived potential perturbations in the model. Recall that in our approach we want to 
consider positive and negative aspects by means of  the two agents $\ang$ and $\dae$. 
 These agents  act over the model by changing the values  of some of the exogenous  components inside the limits of our analysis of reality. 

\begin{definition} 
\label{stress-model}
Let $\ecomodel$ be  a macroeconomic model and let $\mathcal E$ be the set of its exogenous components. 
A \emph{perturbation strength model} for $\mathcal E$ is a set of $\pertur$ of pairs of real numbers, i.e.,   
 $\pertur= \{(\delta_\ang(e),\delta_\dae(e))\mid e\in \mathcal E\}$ describing the {\em potential}  changes that can  be applied to the valuations of the exogenous components  by $\ang$ and $\dae$.   
\end{definition}
Observe that perturbations are  real numbers, so they can be either positive or negative. 
In order to provide intuition on the use of a perturbation strength model let us describe the roles of  $\ang$ and $\dae$ through some examples.

\begin{ex}
\label{perturbation}
A perturbation strength model $\pertur$ for the IS-LM Model $\mathcal E$ given in Example \ref{EasyExample} is:
\[
{\renewcommand{\arraystretch}{1.2}
\begin{tabular}{|| c ||c | c  |c || c |  c | c | c ||}
\hline
agent &$a$ & $b$ & $c$, $d$, $e$,  $f$ & $T$ & $G$ & $M$ & $P$\\
\hline
$\ang$ &$0$ & $+1/20$  & $0$ & $0$ & $+50$ & $0$ & $0$\\
$\dae$ &$0$ & $0$ & $0$ & $+50$ & $-25$ & $0$ & $+1$\\
\hline
\end{tabular}
}
\]
Let us consider the roles of $\ang$ and $\dae$ separately. 
The angel  $\ang$ has the  ability to act upon the parameters  $\{b, G\}$. 
The marginal propensity to consume could be increased from $3/4$ to $4/5$. 
This is modeled by  $\delta_\ang(b)=1/20$.
The government spending  $G$ might  be  increased by $\delta_\ang(G) =50$.
For any other $e\in \mathcal E\setminus \{b, G\}$, $\ang$ has no possibility of  acting upon and therefore $\delta_\ang(e) =0$.
The daemon $\dae$ has the ability to act upon some of the parameters in $\{P,T, G\}$.
The price of goods could increase  by $\delta_\dae(P)=1$;
taxes could increase $\delta_\dae(T)=50$; spending could decrease by $\delta_\dae(G)=-25$.
Again, for   $e\in \mathcal E\setminus \{P,T, G\}$, $\delta_\dae(e) =0$.
In this case, who are $\ang$ and $\dae$?  Are they policy makers?
Clearly $\ang$ cannot be a policy maker because a policy maker cannot influence upon the propensity to consume. 
As we said before,  both $\ang$ and $\dae$ are the two faces of the system analyser trying to give values to the possible "perturbations" of the system. The analyser  also need to precise if the perturbations will act positively (like $\ang$) or negatively (like $\dae$) over the system.
\hfill$\Box$
\end{ex}
\begin{ex}
\label{perturbation-IS-MP}
In this  example $\ang$ and $\dae$ can perturb  the predetermined variables $\mathcal V$ in the IS-MP model given in Example \ref{first-IS-MP}. 
One perturbation strength model $\pertur$ is 
\[
{\renewcommand{\arraystretch}{1.2}
\begin{tabular}{|| c ||c | | c | c | c | c||}
\hline
\rule{0pt}{11pt}
agent &$\alpha$, $\rho$, $\phi$, $\theta_\pi$, $\theta_Y$& $\pi^*$ & $\overline{Y}$ & $\epsilon$ & $v$\\
\hline
$\ang$ &$0$ & $0$ & $+25$ & $+2$ & 0\\
$\dae$ &$0$ & $+3$& $0$& $0$& $+2$ \\
\hline
\end{tabular}
}
\]
In $\pertur$ ,  $\ang$ can potentially act over $\overline Y$ and $\epsilon$.  That is 
$\delta_\ang(\overline{Y}) = 25$ and $\delta_\ang(\epsilon) = 2$ but, for any other $e\in\{\alpha, \rho, \phi, \theta_\pi, \theta_Y,\pi^*, v\}$, $\delta_\ang(e)=0$.
$\dae$ can potentially act over $\pi^*$ and $v$ with 
$\delta_\dae(\pi^*) = 3$ and $\delta_\dae(v) = 2$, all other $\delta_\dae(e)=0$.
\hfill\qed
\end{ex}


%
The concrete valuations of  $\mathcal E$  under  perturbation strength model $\pertur$ depend on the particular selection of  components to be perturbed performed by $\ang$ and $\dae$.
For a set $s\subseteq    \mathcal E$, $\# s$ denotes the number of components in $s$.

%
\begin{definition}
\label{DefDeltaStress}
Consider a model $\ecomodel$ having exogenous components $\mathcal E$  under a perturbation strength model $\pertur$. Given a \emph{joint action} $(a,d)$ with $a,d\subseteq \mathcal E$. The valuation under strength $\mathcal E'$ of $\mathcal E$ given $\pertur$ and $(a,d)$ is noted $\stress_{\pertur}(\mathcal E)[a,d]$  and it is defined as follows.
For any $e\in \mathcal E$, 
$\stress_{\pertur}(e)[a,d] =e+\delta_{\pertur}(e)[a,d]$  where 
\[
\delta_{\pertur}(e)[a,d]=
\begin{cases}
0 & e\notin a\cup d\\
\delta_\ang(e)& e\in a\setminus d\\
\delta_\dae(e) & e\in d\setminus a\\
\delta_\ang(e)+\delta_\dae(e) & e\in a\cap d
\end{cases}
\]%
finally, 
$\stress_{\pertur}(\mathcal E)[a, d]=\{\stress_{\pertur}(e)[a,d]\mid e\in \mathcal E\}$.
\end{definition}

As we said before, there is a convention in empirical work to  measure  
risk and uncertainty by 
historical standard deviation. 
Assume that  $e\in \mathcal E$ is estimated 
statistically  having expectation $\mu_e$ and  
standard deviation $\sigma_e$ \citep{ArgimiroBook}. The values $\mu_e$ and $\sigma_e$, can be used  to define different  perturbation strength models. We focus on  different possible cases (even if they are unlikely). 
Assume that initially we take for the propensity to consume the value $b=\mu_b$. 
If we are interested in modeling a  case  where  consumption is over the mean  $\mu_e$ and  we feel that this is  good (for people), we take $\ang$   
as the agent which can potentially  increase the estimation of the consumption by setting $\delta_\ang(b)=\sigma_b$. In such a case, the modeled system behaves with  a propensity to consume  $b'=\mu_b+\delta_\ang(b)= \mu_e+\sigma_e$.
If we like to model the case where consumption is "slightly" under the mean and we feel that this is not so good, we can take $\delta_\dae(b)=-\sigma_b/2$ and $b'=\mu_e+\delta_\dae(b)=\mu_e-\sigma_b/2$ (assuming $\mu_e>\sigma_b/2$).  Finally,  we could consider a situation in which both cases are possible together into $b'=\mu_e+\delta_\ang(b)+ \delta_\dae(b) = \mu_e+\sigma_e/2$.  

Observe that in the previous  case $\ang$ and $\dae$ are far from being  policy makers. They just provide  a way to focus on  possible system behaviours. However, it is also possible to frame this approach into a policies.
Consider a policy maker liking to know the possible effects of a tax change. Suppose that he is uncertain about the effects of a tax change $\delta_T>0$.  When  the tax 
increase is considered good, we can take $\delta_\ang(T)=\delta_T$ and 
 $\delta_\dae(T)=-\delta_T$.
In the opposite case  (decreasing taxes is a good thing), we can  take  
$\delta_\ang(T)=-\delta_T$ and  $\delta_\dae(T)=\delta_T$. 

Finally, a last important case remains to be considered. 
Imagine an analyser (or a planner) aiming to look at the resilience of a systems under some new and quite open future situation.  Note that this future situation might not  be "the one more likely to happen" (according to the planner opinion). It could be a strange, rather unlikely but extremely dangerous (or extremely favourable)  situation. Therefore   a $\pertur$ model is a way to study the strength of the system under such situation.
Thus, $\pertur$ it might or might not be derived from a statistical study  or a policy plan.
It is precisely  the task of the analyser to frame the future uncertainty into a perturbation strength model providing the potential actions that  $\ang$ and $\dae$ can perform on the parameters.

Let us move to the computation of the equilibrium point in the valuation obtained after a joint action $(a,d)$ of the two agents, denoted  as  
$(Y(\stress_{\pertur}(\mathcal E)[a,d]), r(\stress_{\pertur}(\mathcal E)[a,d])$.
When $\pertur$ is clear from the context, and we want to emphasize the role of the choice of parameters
$(a,d)$ we note $\stress_{\pertur}(\mathcal E)[a, d]$ as $\mathcal E(a,d)$ and the equilibrium point as
$(Y(a,d), r(a,d))$.
The following result points  out some basic properties  between 
the different components:  $\mathcal E$, $\pertur$, $(a,d)$ and $\stress$. In particular, it analyzes explicitly  some cases where parts  of the system (or the whole system)  remains unchanged. In other words, when in $\pertur$ no unattended modifications can appear.

\begin{lemma} 
\label{lemmaTrivial} Let  $\ecomodel$ be a model  having exogenous components $\mathcal E$  under a perturbation strength model $\pertur$ and consider a joint  action $(a,d)$. Then 
$\stress_{\pertur}(e)[a,d] = e$ if and only if either $ e\notin a\cup d$ or  $\delta_\ang(e)=\delta_\dae(e)=0$. The whole system remains unperturbed, i.e. $\stress_{\pertur}(\mathcal E)[a,d] =\mathcal E$ when  $\pertur = \{(0,0)\mid e\in \mathcal E\}$ or
 $(a,d)=(\emptyset,\emptyset)$.

\end{lemma}
%
\begin{proof}
Let $(a,d)$ be a joint  action.
For $e'\in \mathcal E(a,d)$,  according to Definition~\ref{DefDeltaStress},  we have
$e'=e+\delta_{\pertur}(e)[a,d]$. 
Thus,  $e'=e$  if and only if  $\delta_{\pertur}(e)[a,d]=0$ and therefore. The later condition is equivalent to $\delta_\ang(e)=\delta_\dae(e)=0$  or $ e\notin a\cup d$.
The second part of the statement follows trivially from this fact and the definitions. 
\hfill
\end{proof}
%

In the following  example we provide an application of the Definition \ref{DefDeltaStress} to the IS-LM model.
%
%

\begin{ex}
\label{ExStress}
We continue with Examples \ref{Mankiw} and  \ref{perturbation}  under the joint action  $(a, d)=(\{b\}, \{P,G\})$.  After the joint action we get a new valuation,  letting  $\mathcal E'=\mathcal E(\{b\}, \{P,G\})$
and $e'=\stress(e)[\{b\}, \{P,G\}]$. 

We have $\mathcal E'=
\stress_{\pertur}(\mathcal E)[\{b\}, \{P,G\}] =\{a',b',c',d',e',f',T',G',M',P'\}$.
The valuations $\mathcal E$ and the computation  of $\mathcal E'$ is sketched in the following table. 

\begin{footnotesize}
\[
{\renewcommand{\arraystretch}{1.2}
\begin{tabular}{||p{0.5cm} ||p{1.5cm} ||c | p{0.7cm} | p{0.3cm}  |p{0.25cm} | c | c | 
|c | p{0.45cm} | p{0.45cm}| p{0.2cm} ||}
\hline
agent & choice   &$a$ & $b$ & $c$ & $d$ & $e$ & $f$ & $T$ & $G$ & $M$ & $P$\\ 
\hline
 & & $200$ & $3/4$ & $200$ & $25$ & $1$ & $100$ & $100$ & $100$ & $1000$ & $2$\\
\hline
$\ang$ & $a=\{b\}$  & & $+1/20$ & & & & & & & &  \\
$\dae$ & $d=\{P,G\}$ & &  & & & & & &$-25$ & & +1 \\
\hline
&    & $a'$ & $b'$ & $c'$ & $d'$ & $e'$ & $f'$ & $T'$ & $G'$ & $M'$ & $P'$\\
\hline
& & 200 & $4/5$ & 200 & 25 & 1 & 100 & 100 
& $75$
& 1000 &3\\
\hline
\end{tabular}
}
\]
\end{footnotesize}
Observe that  in the joint action $\ang$ acts over the marginal propensity to consume $b$ and $\dae$ acts  (at the same time) over the price index $P$
and the exogenous government spending $G$. 
As $a\cup d= \{b, P,G\}$, 
in the tuple $\mathcal E'=\mathcal E(\{b\}, \{P,G\})$ only the values
corresponding to the components $b$, $P$ and $G$ are perturbed. The other parameters  remain unchanged. According to Definition \ref{DefDeltaStress} and Lemma \ref{lemmaTrivial}, $e=e'$ for $e\in \mathcal E\setminus\{b,P,G\}$ and $e'= e+\delta_{\pertur}(e)[a,d]$ for 
$e\in \{b,P,G\}$. Taking into account the perturbation strength model $\cal S$ we get
\begin{align*}
b' &=b+\delta_\ang(b)=3/4+1/20=4/5, \\
P' &=P+\delta_\dae(P) = 2+1=3,\\
G' &=G+\delta_\dae(G) = 100-25=75.
\end{align*} 
\remove{Let us compute $b'=b+ \delta_{\pertur}(b)[\{b\}, \{P,G\}]$.
Note that $b'$,  the value of the marginal propensity to consume $b$ when perturbed according to the selected joint action.  
As $b$ is perturbed  only by $\ang$ as we can see in the equalities  
$b\in a\setminus d=\{b\}\setminus\{P,G\}=\{b\}$. Therefore, according to Definition \ref{DefDeltaStress} we have $\delta_{\pertur}(b)[\{b\}, \{P,G\}]=\delta_\ang(b)=1/20$  and we  get
$b'=b+\delta_\ang(b)=3/4+1/20=4/5$. 

Let compute $P'$. Note that $P$ is just perturbed by $\dae$ because 
$P\in \{P,G\}\setminus \{b\}$. As $\delta_{\pertur}(P)[\{b\}, \{P,G\}]= \delta_\dae(P)= +1$ we obtain
$P'=P+\delta_\dae(P) = 2+1=3$. Similarly,  as 
$\delta_\dae(G)= -25$ we get
$G'=G+\delta_\dae(G) = 100-25=75$. 
}
Finally, the  equilibrium point corresponding to $\mathcal E(\{b\}, \{P,G\})$ is 
$Y(\{b\}, \{P,G\})= 28700/27\approx 1062.96$, 
$r(\{b\}, \{P,G\})=197/27\approx  7.29$.
\hfill$\Box$
\end{ex}
When $\ecomodel$  is  perturbed from valuation $\mathcal E$ into $\mathcal E'$ by joint action $(a,d)$, we would like to isolate the effects of the the perturbation by expressing the equilibrium point $(Y(a,d), r(a,d))$  with respect to the non-perturbed  equilibrium point $(Y,r)$. 
\remove{
Informally  $\ecomodel' = \ecomodel + \cdots$ where the dots corresponds to the perturbation terms. 
Such factorizations are needed to obtain the results given in Sections 
\ref{section-IS-LM} and \ref{section-IS-MP}.
In the case of the 
 IS-LM or IS-MP we would like to get factorizations: 
\[
\left( 
\begin{array}{c}
Y(a,d)\\
r(a,d)
\end{array}
\right)=
\left( 
\begin{array}{c}
Y\\
r
\end{array}
\right)+ \cdots
, \quad
\left( 
\begin{array}{c}
Y(a,d)\\
\pi(a,d)
\end{array}
\right)=
\left( 
\begin{array}{c}
Y\\
\pi
\end{array}
\right)
+
\cdots
\]}
In general this is difficult to obtain, however, when an important  part of $\mathcal E$ cannot be perturbed, we can take advantage of the linear structure of the IS-LM and IS-MP models and get an explicit expression. In the following lemmas we provide such expressions for some of such cases. Those results will be used later on.  

\remove{ of deals with such a cases. Lemmas  \ref{linearity-IS-LM} and 
\ref{EasyFiscalPolicies} deals with the  IS-LM model. Lemmas \ref{IS-MP:lemma-linear} and \ref{JustPerturbationsIS-MP} 
deals with the IS-MP model.}
\begin{lemma}
\label{linearity-IS-LM}
Consider a perturbation strength model $\pertur$ for the IS-LM model such that, for 
$e\in\{b, d, e ,f, P\}$, we have   $\delta_\ang(e)=\delta_\dae(e) =0$.  
Let  $(a,d)$ be a joint action and define 
\[\delta_{\pertur}(a,c,G, T)[a,d] =
\delta_{\pertur}(a)[a,d] +\delta_{\pertur}(c)[a,d]+\delta_{\pertur}(G)[a,d]-b\, \delta_{\pertur}(T)[a,d].
\]
Then,  it holds 
\begin{eqnarray*}
\left( 
\begin{array}{c}
Y(a,d)\\
r(a,d)
\end{array}
\right)=
\left( 
\begin{array}{c}
Y\\
r
\end{array}
\right)+
\frac{1}{g}\left(
\begin{array}{cc}
f & d/P\\
e &\hbox{\;\;}-(1-b)/P
\end{array}
\right)
\left(
\begin{array}{c}
\delta_{\pertur}(a,c,G,T)[a,d]\\
\delta_{\pertur}(M)[a,d]
\end{array}
\right) 
\end{eqnarray*}
\end{lemma}

\begin{proof}
As $\ang$ and $\dae$ cannot perturb  components in $\{b,d,e, f, P\}$,  the $2\times 2$ matrix and the value $g=(1-b)f+de$ remain unchanged under $\pertur$ and any joint action $(a,d)$. Therefore,

\begin{eqnarray*}
\left( 
\begin{array}{c}
Y(a,d)\\
r(a,d)
\end{array}
\right)=
\frac{1}{g}\left(
\begin{array}{cc}
f & d/P\\
e &\hbox{\;\;}-(1-b)/P
\end{array}
\right)
\left(
\begin{array}{c}
a'+c'+G'-bT'\\
M'
\end{array}
\right) 
\end{eqnarray*}
and  we get
\begin{eqnarray*}
&&a'+c'+G'-b\, T'\\
&&\qquad= a+c+G-b\, T+ \delta_{\pertur}(a)[a,d] +\delta_{\pertur}(c)[a,d]+\delta_{\pertur}(G)[a,d]-b\, \delta_{\pertur}(T)[a,d]\\
&&\qquad= a+c+G-b\, T +\delta_{\pertur}(a,c,G,T)[a,d],\\
&&M'=M+\delta_{\pertur}(M)[a,d].
\end{eqnarray*}
Using straightforward linear algebra the result follows.
\end{proof}

We are also interested in valuations $\mathcal E$ in relation to fiscal policies. In those situations only $T$ and $G$ can suffer a perturbation. In such a case we get the following result.

\begin{lemma}
\label{EasyFiscalPolicies}
Consider a perturbation strength model  $\pertur$ where $\{G,T\}$ are the unique components that can be perturbed.   
For a    joint  action  $(a,d)$,  we have 
\begin{eqnarray*}
\left( 
\begin{array}{c}
Y(a,d)\\
r(a,d)
\end{array}
\right)=
\left( 
\begin{array}{c}
Y\\
r
\end{array}
\right)
+
\frac{1}{g}\delta_{\pertur}(G,T)[a,d] 
\left(
\begin{array}{c}
f \\
e 
\end{array}
\right)
\end{eqnarray*}
where $\delta_{\pertur}(G,T)[a,d] =
\delta_{\pertur}(G)[a,d]-b\delta_{\pertur}(T)[a,d]$.
\end{lemma}

\begin{proof}
As  $e\in \{a,b,c,d,e,f,M,P\}$ cannot be perturbed, i.e.,
$\delta_\ang(e)=\delta_\dae(e)=0$.  From  Lemma \ref{linearity-IS-LM}, we have 
$\delta_{\pertur}(M)[a,d]=0$. Therefore, 
\[
\delta_{\pertur}(a,c,G,T)[a,d] =
\delta_{\pertur}(G)[a,d]-b\, \delta_{\pertur}(T)[a,d] =\delta_{\pertur}(G,T)[a,d] 
\]
and by computing the matrix product the result follows.
\end{proof}

%

Our next result provides the equilibrium point in the  IS-MP model when the exogenous parameters cannot be perturbed.

\begin{lemma}
\label{IS-MP:lemma-linear}
Consider a perturbation strength model $\pertur$ for the IS-MP model such that, for 
$e\in\mathcal P$,  $\delta_\ang(e)=\delta_\dae(e)=0$.  For any joint  action 
$(a,d)$ it holds that
\[
\left( 
\begin{array}{c}
Y(a,d)\\
\pi(a,d)
\end{array}
\right)=
\left( 
\begin{array}{c}
Y\\
\pi
\end{array}
\right)
+
\left( 
\begin{array}{rr}
\hat\gamma &\hbox{\;\;}-\hat\delta\\
\hat \rho & \hat\mu
\end{array}
\right)
\left( 
\begin{array}{c}
\delta_{\pertur}(\epsilon)[a,d]\\
\delta_{\pertur}(v)[a,d]
\end{array}
\right)+
\left( 
\begin{array}{c}
\delta_{\pertur}(\overline{Y})[a,d]\\
\delta_{\pertur}(\pi^*)[a,d]
\end{array}
\right).
\]   
\end{lemma}

\begin{proof}
As $\delta_{\pertur}(e) = 0$, for $e\in \{\alpha,\rho,\phi,\theta_\pi,\theta_Y\}$, $\delta_{\pertur}(e)[a,d]= e$. Therefore, 
the $2\times2$ matrix  in the expression for the equilibrium remains unchanged. Thus the values 
$Y(a,d)$ and $\pi(a,d)$ can be expressed as
\[
\left( 
\begin{array}{c}
Y(a,d)\\
\pi(a,d)
\end{array}
\right)
=
\left( 
\begin{array}{rr}
\hat\gamma &\hbox{\;\;}-\hat\delta\\
\hat \rho & \hat\mu
\end{array}
\right)
\left( 
\begin{array}{c}
\epsilon+\delta_{\pertur}(\epsilon)[a,d]\\
v+\delta_{\pertur}(v)[a,d]
\end{array}
\right)+
\left( 
\begin{array}{c}
\overline{Y} +\delta_{\pertur}(\overline{Y})[a,d]\\
\pi^*+\delta_{\pertur}(\pi^*)[a,d]
\end{array}
\right)
\]   
and, by linearity, the result follows.
\end{proof}

In Section \ref{section-IS-MP}  we will  consider the case where only the income $Y$ and the inflation $\pi$  become uncertain.  
For such a case we have the following expression for the perturbed  equilibrium point. 

\begin{lemma}
\label{JustPerturbationsIS-MP}
Consider a perturbation strength model $\pertur$ for the IS-MP model such that, for 
$e\in \mathcal E\setminus\{\overline{Y}, \pi^*\}$, we have  $\delta_\ang(e)=\delta_\dae(e)=0$. For any
joint action $(a,d)$, it holds
$Y(a,d)= Y+\delta_{\pertur}(\overline{Y})[a,d]$ and 
$\pi(a,d)= \pi+\delta_{\pertur}(\pi^*)[a,d]$.
\end{lemma}

\begin{proof}
Observe that   $ \mathcal E\setminus\{\overline{Y}, \pi^*\}= \mathcal P \cup\{\epsilon,v\}$.
As, for  $e\in \mathcal P$, $\delta_\ang(e)=\delta_\dae(e)=0$ we can use  
Lemma~\ref{IS-MP:lemma-linear}. Thus, we get
\[
\left( 
\begin{array}{c}
Y(a,d)\\
\pi(a,d)
\end{array}
\right)=
\left( 
\begin{array}{c}
Y\\
\pi
\end{array}
\right)
+
\left( 
\begin{array}{rr}
\hat\gamma &\hbox{\;\;}-\hat\delta\\
\hat \rho & \hat\mu
\end{array}
\right)
\left( 
\begin{array}{c}
\delta_{\pertur}(\epsilon)[a,d]\\
\delta_{\pertur}(v)[a,d]
\end{array}
\right)+
\left( 
\begin{array}{c}
\delta_{\pertur}(\overline{Y})[a,d]\\
\delta_{\pertur}(\pi^*)[a,d]
\end{array}
\right)
\]   
As  $\delta_\ang(e)=\delta_\dae(e)=0$,  for any $e\in \{\epsilon, v\}$, we have $\delta_{\pertur}(\epsilon)[a,d]=\delta_{\pertur}(v)[a,d]=0$ and the claimed result follows.
\end{proof}

We conclude this section with an example illustrating the details of the computation of the equilibrium point in a perturbed scenario.  

\begin{ex}
\label{ex-stress-IS-MP}
Let us continue with the Example \ref{perturbation-IS-MP} in the IS-MP model. Consider the joint action
$(a,d)=(\{\epsilon\},\{\pi^*,v\})$.
Writing the new values  as 
$e'=\stress(e)[(\{\epsilon\},\{\pi^*,v\}]$. As usual, 
\[\mathcal E'=\mathcal E(\{\epsilon\},\{\pi^*,v\})=\{\alpha',\rho',\phi',\theta'_\pi,\theta'_Y,{\pi^*}', \overline{Y}',\epsilon',v'\}.\]
A sketch of the computation of the perturbed valuation is given in the following table. 
\begin{footnotesize}
\[
{\renewcommand{\arraystretch}{1.2}
\begin{tabular}{||p{0.5cm} ||p{1.5cm} ||p{0.3cm}|p{0.3cm}| p{0.5cm} | p{0.5cm}  |p{0.5cm} || c | c | 
c | p{0.5cm} ||}
\hline
agent & choice   &$\alpha$ & $\rho$ &$\phi$& $\theta_\pi$ & $\theta_Y$ & $\pi^*$ & \rule{0pt}{11pt}$\overline{Y}$ & $\epsilon$ & $v$ \\ 
\hline
 & & $1$ & $2$ & $1/4$ & $1/2$ & $1/2$ & $2$ & $100$ & $1$ &$1/2$ \\
\hline
$\ang$ & $a=\{\epsilon\}$  & &  & & & & &  &$+2$ & \\
$\dae$ & $d=\{\pi^*, v\}$ & &  & & & &$+3$ & &  &$+2$\\
\hline
&    & $\alpha'$ & $\rho'$ & $\phi'$ & $\theta'_\pi$ & $\theta_Y'$ & ${\pi^*}'$ & \rule{0pt}{11pt}$\overline{Y}'$ & $\epsilon'$ &$v'$\\
\hline
&&$1$ & $2$ & $1/4$ & $1/2$ & $1/2$ &5&100&$3$&$5/2$\\
\hline
\end{tabular}
}
\]
\end{footnotesize}

As $\epsilon\in a\setminus d$, according to the Definition \ref{DefDeltaStress} we have
\[\epsilon'=\stress(\epsilon)[(\{\epsilon\},\{\pi^*,v\}]= \epsilon+\delta_{\pertur}[\{\epsilon\},\{\pi^*,v\}]= \epsilon+\delta_\ang(\epsilon) = 3.\]
As $\pi\in d\setminus a$, 
${\pi^*}' = \pi+\delta_\dae(\pi^*) = 5$.  Similarly $v'=v+\delta_\dae(v) = 5/2$. All other values remain unchanged.
In order to obtain  $(Y',{\pi^*}')=(Y(a,d),\pi(a,d))$, by using  Lemma \ref{IS-MP:lemma-linear}, we have
 \[
\left( 
\begin{array}{c}
Y'\\
\pi'
\end{array}
\right)=
\frac{1}{13}
\left( 
\begin{array}{c}
1306\\
34
\end{array}
\right)
+
\frac{1}{13}
\left( 
\begin{array}{rr}
8 &\hbox{\;\;}-4\\
2 & 12
\end{array}
\right)
\left( 
\begin{array}{c}
\delta_\ang(\epsilon)\\
\delta_\dae(v)
\end{array}
\right)+
\left( 
\begin{array}{c}
0\\
\delta_\ang(\pi^*)
\end{array}
\right)
=
\frac{1}{13}
\left( 
\begin{array}{c}
1314\\
101
\end{array}
\right).
\]   
\hfill
\qed
\end{ex}
%


\section{Uncertainty Profiles and \angdae games}
\label{Uncertainty}
In this section, we present the model and tools to analyse a situation under  a given perturbation strength model. Recall that the perturbation strength model fixes  the perturbation to the different exogenous components.
Following \citet{GSS2014-CJ}, we introduce \emph{uncertainty profiles}  as a tool to  provide
an \emph{a priori} (global and macroscopic) view of the macroeconomic  model under perturbation. An uncertainty profile describes a priori situation where we are uncertain over the specific location where the perturbation will impact the system but we have an approximate idea of the extension of the perturbation. Uncertainty profiles are based in three components. The first one  identifies  the   set of exogenous components  that might be  perturbed. The second states the limits in the number of components that can suffer perturbation. The third component quantifies (as function of the exogenous components)  the benefits for the agents. In such a setting,  we are uncertain about the specific subset that will  suffer the perturbation.  
\begin{definition}
\label{uncertainty-macro}
Let $\ecomodel$  be a macroeconomic model having  $\mathcal E$ as  set of exogenous components. Let $S$ be a perturbation strength model for $\ecomodel$.  An  \emph{uncertainty profile}  is a tuple 
$\mathcal U=\langle \mathcal E, \pertur, \mathcal A,\mathcal D,  b_{\ang}, 
b_{\dae}, u_{\ang}, u_{\dae} \rangle$ where
$\mathcal A, \mathcal D \subseteq \mathcal E$. 
The \emph{spread} of the  perturbation  
$b_{\ang}$ and $b_{\dae}$, verify 
$b_{\ang}\le \#\mathcal A$ and
similarly $b_{\dae}\le \#\mathcal D$. 
The exerted perturbation follows from  joint  actions 
$(a,d)$ verifying   
 $a\subseteq \mathcal A$,   $d\subseteq \mathcal D$  with $\# a =b_{\ang}$ and $\# d =b_{\dae}$.
The effects of a joint action  are  measured by  the \emph{utility functions} $u_{\ang}$ and $u_{\dae}$.  Given a joint  action  $(a,d)$,   $u_{\ang}(a,d)$ measures  $\ang$\rq{}s gain while  $u_{\dae}(a,d)$ measures  $\dae$\rq{}s  gain. 
\end{definition}

The analyser has the perception that, when an angelic component belonging to  $\mathcal A$ is perturbed,   it is unlikely to have a malicious  impact.
In contrast, when a daemonic component in $\mathcal D$ is perturbed,  it might well have catastrophic implications. Note that, both $\mathcal A$ and $\mathcal D$ determine the potential parameters to be perturbed. In many cases the analyser do not expect that all of them are perturbed  "at the same time".
Thus, $b_{\ang}$ gives the number of components in $\mathcal A$ that can be perturbed together.  In a similar way $b_{\dae}$ gives the limitations for the $\dae$.
  
Uncertainty profiles provide a flexible analysis tool. Let us consider some cases of interest. 
A really pessimistic and global approach like \rq\rq{}anything can go wrong\rq\rq{} can be modelled taking 
$\mathcal A=\emptyset$ and $\mathcal D = \mathcal E$. Pessimism increases by selecting  bigger values for
$b_{\dae}$. The Murphy's law, 
"anything that can go wrong will go wrong" translates directly as  $b_{\dae}=\# \mathcal D$. When the perturbation of some component  can  go well or go wrong but not both at the same time,  it is enough to assume that $\mathcal A\cap \mathcal D=\emptyset$. When we are uncertain about a component  $e$ and we presume that it could suffer positive and negative influences at the same time, we can take 
$e\in \mathcal A\cap \mathcal D$. A completely optimistic perception  can be modelled through 
$\mathcal A= \mathcal E$, $\mathcal D=\emptyset$. Here the degree of optimism is given by the value of $b_ {\ang}$. 
Setting $\mathcal A=\mathcal D = \mathcal E$ models a situation of maximal uncertainty, the different degrees of perturbation,
on both sides,  are tuned trough the values $b_{\ang}$ and  $b_{\dae}$. 
For instance, 
$b_{\ang} = \#\mathcal E$ and $b_{\dae} = \#\mathcal E$
 forces the $\ang$ and the $\dae$ to act perturbing simultaneously all the components in  $\mathcal E$.

\remove{To isolate the possible joint perturbed actions we define the pairs $(a,d)$ such that
$
(a,d)\in \{a\subseteq {\mathcal A}\;|\; \#a=b_{\ang}\}\times 
\{d\subseteq {\mathcal D}\;|\; \#d=b_{\dae}\}.
$}
The situation described by  an uncertainty profile $\mathcal U$ is analyzed by means of an  associated strategic \angdae game. In such a game $\ang$ and $\dae$ decide their actions strategically.  
For basics on game theory we refer the reader to \citep{Osborne2004,OsRu94}). 
\remove{Given  ${\cal  U}$
a strategic situation 
arises when the model $\ecomodel$ is subjected to combined  angel and daemon 
actions. }
\begin{definition}
Let $\ecomodel$  be a macroeconomic model having  $\mathcal E$ as  set of exogenous components. Let  $\pertur$  be a perturbation strength model for $\ecomodel$. Let  $\mathcal U=\langle \mathcal E, \pertur, \mathcal A,\mathcal D,  b_{\ang}, 
b_{\dae}, u_{\ang}, u_{\dae} \rangle$ be an uncertainty profile.
The associated \emph{angel-daemon} strategic game (the \angdae game) $\Gamma(\mathcal U)$ is defined as $\Gamma(\mathcal U) = \langle \{\ang,\dae\},A_\ang, A_\dae, u_\ang,u_\dae \rangle$. $\Gamma(\mathcal U) $ has two players $\{\ang, \dae\}$. The player\rq{}s actions are 
$A_\ang=\{a\subseteq {\mathcal A}\;|\; \#a=b_{\ang}\}$ and 
$A_\dae=\{d\subseteq {\mathcal D}\;|\; \#d=b_{\dae}\}$.
Their  utilities  are $u_\ang$ and $u_\dae$.
\end{definition}

Notice that, in an \angdae game the set of strategy profiles is $A_\ang\times A_\dae$. Thus strategy profiles are permissible joint  actions. 
Recall that, a \emph{pure Nash equilibrium}  is a strategy profile  such that neither $\ang$ nor $\dae$ can improve the situation by himself \citep{Osborne2004}.  Formally, a strategy profile $(a,d)$ 
is a pure Nash equilibrium if and only if  $u_\ang(a,d)\ge u_\ang (a',d)$, for all 
$a'\in A_\ang$, and $u_\dae(a,d)\ge u_\dae(a,d')$, for all $d'\in A_\dae$. We note by $\PNE(\Gamma)$  the set of pure Nash equilibria of $\Gamma$. When $\Gamma$ is clear from the context we just write $\PNE$.
Given $d\in A_d$ the {\em best response} of $\ang$ to  $\dae$\rq{}s choice $d$ is the set of 
strategies giving to $\ang$ the maximum utility, i.e.,
$B_\ang(d)=  \{a\mid u_\ang(a,d)\ge u_\ang(a',d) \hbox{\; for all\;} a'\in A_\ang\}$.
Similarly, $B_\dae(a)=  \{d\mid u_\dae(a,d)\ge u_\dae(a,d') \hbox{\; for all\;} d'\in A_\dae\}$. It is well known that 
$(a,d)\in \PNE$ if and only if $a\in B_\ang(d)$ and $b\in B_\dae(a)$. Formally,
$\PNE = \{(a,d)\mid a\in B_\ang(d) \hbox{\; and\;} b\in B_\dae(a)\}$.  

In the following we analyze the existence of PNE for some extreme types of uncertainty profiles. 
\begin{lemma}
\label{LemmaNoFreedom}
Let $\mathcal U=\langle \mathcal E, \pertur, \mathcal A,\mathcal D,b_{\ang},b_{\dae}, u_{\ang},u_{\dae}\rangle$ be an uncertainty profile for $\ecomodel$.   
When $b_{\ang} =\#\mathcal A$ and $b_{\dae} =\#\mathcal D$,
the only \PNE of  $\Gamma(\mathcal U)$ is  $(\mathcal A,\mathcal D)$.
When $b_{\ang} = 0$ and $b_{\dae} =0$,
the only \PNE of  $\Gamma(\mathcal U)$ is  
is  $(\emptyset,\emptyset)$.
\end{lemma}

\begin{proof}
In the first case, the only permissible  joint action is $(\mathcal A, \mathcal D)$,  In the second case  $(\emptyset,\emptyset)$ is the unique permissible joint strategy. Therefore, in both cases, there is only one possible strategy profile  and the result follows. 
\end{proof}

The previous lemma considers two extreme cases where neither $\ang$ nor $\dae$ have freedom to make choices. 
Our next result shows that in an equilibrium,  when one of the two agents cannot act,  the other agent maximizes his utility.

\begin{lemma}
Let $\mathcal U=\langle \mathcal E, \pertur, \mathcal A,\mathcal D,b_{\ang},b_{\dae}, u_{\ang},u_{\dae}\rangle$ be an uncertainty profile for $\ecomodel$.   
When $b_{\dae} =0$ it holds
$\PNE(\Gamma(\mathcal U)) = \{(a,\emptyset)\mid u_{\ang}(a,\emptyset)= \max_{a'\in A_\ang}u_{\ang}(a',\emptyset)\}$.
When $b_{\ang} =0$
it holds
$\PNE(\Gamma(\mathcal U)) = \{(\emptyset,d)\mid u_{\dae}(\emptyset, d)= \max_{d'\in A_\dae}u_{\dae}(\emptyset,d')\}$.
\end{lemma}

\begin{proof}
In the first case we have $A_\ang\times A_\dae=\{(a,\emptyset)\mid a\subseteq \mathcal A \hbox{\;and \;} \#a=b_{\ang} \}$. Symmetrically, in the second case we have  
$A_\ang\times A_\dae=\{(\emptyset,d)\mid d\subseteq \mathcal D \hbox{\;and \;} \#d=b_{\dae} \}$. The claim follows from the definition of \PNE.
\end{proof}

Let us present some examples of \angdae games arising in the analysis of uncertainty profiles for the  IS-LM and the IS-MP models. We have selected them to illustrate some aspects related to the existence of PNE in the associated \angdae game.


\begin{ex}
\label{FirstGame}
Take $\ecomodel=\hbox{IS-LM}$,  the valuation  
$\mathcal E$ given in Example \ref{EasyExample}
and the perturbation strength model  $\pertur$ introduced in Example \ref{perturbation}.
We will define $\mathcal A$ and $\mathcal D$ depending on the uncertain scenario we are interested to describe. 
Consider  a situation in which  we are interested to know how perturbations affects  $Y$ and  $r$ when the marginal propensity to consume  might be perturbed in an angelic way  while the price of goods and taxes might be perturbed in a daemonic way. However, the exogenous government spending might be perturbed in both directions. 
In consequence we set  
$\mathcal A= \{b,G\}$ and $\mathcal D= \{P,G, T\}$.
Assume that  we do not expect perturbations to be exerted at the same time on more than one component. So, we set $b_{\ang}=b_{\dae}=1$. 
We are interested to know how perturbations affects  $Y$ and  $r$. So, we take   $u_{\ang} = Y$ and  we need to precise the interests of $\dae$ with respect to $r$.
In general, a low interest rate $r$ is a mean to achieve higher economic growth and lower unemployment. Now we explore two cases.

First, we want  to catch a ``worst-case'' situation where the perturbation  increases the interest rate.  The dictum, ``having hight interest rate is bad" is captured by setting
$u_{\dae} = r$.  In this case, $\dae$ tries to maximize $r$. The uncertainty profile $\mathcal U_1 = \langle \mathcal E,  \pertur,\{b,G\}, \{P,G, T\}, 1, 2, Y, r\rangle$
mimics an uncertain situation asking  at the same time for a hight income and a high interest rate. 
Second, we consider a ``best-case'' situation. As raising interest rate $r$ has negative effects, we are  interested to know what happens when this is not the case. So, $r$ is considered a dis-utility.  Raising $r$ is seen as a negative fact. Thus we  define   $u_{\dae} = -r$ (note that maximizing $-r$ is the same as minimizing $r$). These considerations lead to the  uncertainty profile 
$\mathcal U_2 = \langle \mathcal E,  \pertur,\{b,G\}, \{P,G, T\}, 1, 2, Y, -r\rangle$.

Let us consider the first case, $\mathcal U_1 = \langle \mathcal E,  \pertur,\{b,G\}, \{P,G, T\}, 1, 2, Y, r\rangle$ and analyze the PNE of the  \angdae game
$\Gamma(\mathcal U_1)$. According to the definitions, the set of actions for the players are  
\begin{eqnarray*}
&&A_\ang = 
\{a\subseteq \{b,G\}\;|\; \#a=1\}=\{\{b\},\{G\}\},\\
&&A_\dae = \{d\subseteq \{P,G, T\}\;|\; \#d=2\}=
\{\{P,G\},\{P,T\},\{T,G\}\}
\end{eqnarray*}
The utilities are $u_\ang = Y$ and $u_\dae= r$ which can be tabulated in the usual bi-matrix form. Recall that a \emph{bi-matrix} is a table with an entry for each strategy profile holding  a pair of values corresponding to the utilities of the two players. 
Computing the utilities of $\ang$ and $\dae$ (according to the methods developed in Example \ref{ExStress}), after applying the perturbation to the selected components,  we get, for example 
$Y(\{b\}, \{P,G\}) \approx 1062.96$
and  $r(\{b\}, \{P,G\})\approx 7.29$. 
The remaining results are summarized in the following bi-matrix representation of $\Gamma(\mathcal U_1)$.

\medskip
{
{\begin{tabular}{cc|c|c|c|}
\cline{3-5}
& &\multicolumn{3}{|c| }{$\dae$}\\
\cline{3-5}
   & & $\{P,G\}$    & $\{P,T\}$ & $\{T,G\}$\\
\cline{1-5}
\multicolumn{1}{|c|}{\multirow{2}{*} {$\ang$}} &
\multicolumn{1}{|c|}{$\{b\}$}   
& $1062.96, 7.29$
& $1029.62, 6.962$
& $1233.33, 22/3\approx 7.33$\\
\multicolumn{1}{|c|}{}   &
\multicolumn{1}{|c|}{$\{G\}$}   
& $1066.66, 22/3\approx7.33$
& $1041.66, 7.08$
& $1075, 5.75$\\
\cline{1-5} 
\end{tabular}}
}

\bigskip
The sets of  $\ang$\rq{}s best responses to $\dae$\rq{}s actions are:
\[B_\ang(\{P,G\})=\{\{G\}\},
B_\ang(\{P,T\}) = \{\{G\}\},
B_\ang(\{T,G\})=\{\{b\}\}\]
and the sets of $\dae$\rq{}s best responses to $\ang$\rq{}s actions are:
\[B_\dae(\{b\}) = \{\{T,G\}\}, 
B_\dae(\{G\}) = \{\{P,G\}\}
\]
The strategy profile $(\{G\},\{P,G\})$ is a pure Nash equilibrium 
because $\{G\}\in B_\ang(\{P,G\})$ and $\{P,G\}\in B_\dae(\{G\})$.
Similarly, $(\{b\},\{T,G\})\in PNE$.  So, we have 
\[\PNE=\{(\{G\},\{P,G\}), (\{b\},\{T,G\})\}\] 
The utilities for  $(\{G\},\{P,G\})$ are
\begin{eqnarray*}
&&u_{\ang}(\{G\},\{P,G\})=  Y(\{G\},\{P,G\})= 1066.66\\
&&u_{\dae}(\{G\},\{P,G\}) = r(\{G\},\{P,G\})= 22/3\approx 7.33
\end{eqnarray*}
and the utilities for $(\{b\},\{T,G\})$ are 
\begin{eqnarray*}
&&u_{\ang}(\{b\},\{T,G\})=  Y(\{b\},\{T,G\})= 1233.33\\
&&u_{\dae}(\{b\},\{T,G\})= r(\{b\},\{T,G\})= 22/3\approx 7.33
\end{eqnarray*}
Thus, $\Gamma(\mathcal{U}_1)$, has  more than one PNE. Observe that $\ang$ gets different rewards on these PNE.  


In the second case we have 
$\mathcal U_2 = \langle \mathcal E,  \pertur,\{b,G\}, \{P,G, T\}, 1, 2, Y, -r\rangle$. The corresponding \angdae game
$\Gamma(\mathcal U_2)$ is described by the following bi-matrix form. 

\medskip
{\begin{tabular}{cc|c|c|c|}
\cline{3-5}
& &\multicolumn{3}{|c| }{$\dae$}\\
\cline{3-5}
   & & $\{P,G\}$    & $\{P,T\}$ & $\{T,G\}$\\
\cline{1-5}
\multicolumn{1}{|c|}{\multirow{2}{*} {$\ang$}} &
\multicolumn{1}{|c|}{$\{b\}$}   
& $1062.96, -7.29$
& $1029.62, -6.962$
& $1233.33, -22/3\approx -7.33$\\
\multicolumn{1}{|c|}{}   &
\multicolumn{1}{|c|}{$\{G\}$}   
& $1066.66, -22/3\approx-7.33$
& $1041.66, -7.08$
& $1075, -5.75$\\
\cline{1-5} 
\end{tabular}
}

\bigskip
As $\ang$ tries to maximize the outcome $Y$, their sets of  best responses are the same as in 
$\Gamma(\mathcal U_1)$. As $\dae$ tries to maximize $-r$ (minimize $r$), their best responses are $B_\dae(\{b\}) = \{\{P,T\}\}$ and 
$B_\dae(\{G\}) = \{\{T,G\}\}$. Observe that in this case $\PNE=\emptyset$. 
As there is no \PNE, on any joint strategy, one of the agents (or both), can unilaterally improve his situation. For instance, suppose that 
initially the system is in  $(\{b\},\{P,G\})$. In this case, $\dae$ has an incentive to change  strategy 
as  $B_\dae(\{b\}) = \{\{P,T\}\}$. Thus,  moving to $(\{b\},\{P,T\})$. This is denoted as $(\{b\},\{P,G\})\buildrel \dae\over \longrightarrow
(\{b\},\{P,T\})$. Then $\ang$ can improve also his situation. This process  
give rise to an unstable never-ending behaviour.
\begin{eqnarray*}
&&(\{b\},\{P,G\})\buildrel \dae\over \longrightarrow
(\{b\},\{P,T\}) \buildrel \ang\over \longrightarrow
(\{G\},\{P,G\})\buildrel \dae\over \longrightarrow\\
&&\qquad (\{G\},\{T,G\})\buildrel \ang\over \longrightarrow
(\{b\},\{T,G\})\buildrel \dae\over \longrightarrow
(\{b\},\{P,T\}) \buildrel \ang\over \longrightarrow \cdots
\end{eqnarray*}
We have considered two cases dealing with opposite interests on $r$. Many other cases might be of interest, for example when $\ang$ tries to minimize $Y$ and $\dae$ tries to maximize 
$r$. In any case, the analyser should transform a perception about uncertainty into one (or several) uncertainty profiles. 
\hfill$\Box$
\end{ex}
\begin{ex}
\label{IS-MP-game}
Let us continue with the Examples \ref{perturbation-IS-MP} and \ref{ex-stress-IS-MP}  of the  IS-MP model.  We  
consider the uncertainty profiles where $u_{\ang}=Y$ and analyze the situations in which $\dae$ has opposite interests in the  inflation $\pi$.
We consider again two  uncertainty profiles in which the agents have a narrow spread of perturbation: $\mathcal U_1 = \langle \mathcal E,\pertur,  \{\epsilon, \overline{Y}\}, 
\{v,\pi^*\}, 1,1, Y, \pi\rangle$ and 
$\mathcal U_2 = \langle \mathcal E,\pertur,  \{\epsilon, \overline{Y}\}, 
\{v,\pi^*\}, 1,1, Y, -\pi\rangle$.
 
\remove{$\mathcal U = \langle \mathcal E,\pertur,  \{\epsilon, \overline{Y}\}, 
\{v,\pi^*\}, 1,1, Y, \dots\rangle$. Such profiles model  situations in which $\ang$ tries to maximize the income $Y$ (high $Y$ is good). Then $u_{\ang}=Y$.
We associate the inflation $\pi$ with $\dae$. 
We can assume that  high inflation is bad, but this is  not necessarily the case. 
Of course, galloping inflation makes the economy less stable. 
However, an increase in the rate of inflation leads to a decrease in 
the real cost of capital, thus stimulating borrowing and investment. We deal with both cases separately. Taking 
$\mathcal U_1 = \langle \mathcal E,\pertur,  \{\epsilon, \overline{Y}\}, 
\{v,\pi^*\}, 1,1, Y, \pi\rangle$ we deal with the case where high inflation is bad ($\dae$ tries to maximize $\pi$). In
$\mathcal U_2 = \langle \mathcal E,\pertur,  \{\epsilon, \overline{Y}\}, 
\{v,\pi^*\}, 1,1, Y, -\pi\rangle$, the inflation $\pi$ is a dis-utility ($\dae$ tries to minimize $\pi$).}

Computing the utilities as in Example \ref{ex-stress-IS-MP}, 
for instance $Y(\{\epsilon\},\{v\}) =13014/13$ and $\pi(\{\epsilon\},\{v\}) =62/13$, we get the following bi-matrix form for  $\Gamma (\mathcal U_1)$.

\medskip
{\renewcommand{\arraystretch}{1.2}
\begin{tabular}{cc|c|c|}
\cline{3-4}
& &\multicolumn{2}{|c| }{$\dae$}\\
\cline{3-4}
   & & $\{v\}$    & $\{\pi^*\}$\\
\cline{1-4}
\multicolumn{1}{|c|}{\multirow{2}{*} {$\ang$}} &
\multicolumn{1}{|c|}{$\{\epsilon\}$}   
& $1314/13, 62/13$
& $1322/13,77/13$\\
\cline{3-4}
\multicolumn{1}{|c|}{}   &
\multicolumn{1}{|c|}{$\{\overline{Y}\}$}   
& $1623/13, 58/13$
& $1631/13,73/13$\\
\cline{1-4}
\end{tabular}
}

\medskip
As before,  we compute  the sets of best responses:
$B_\ang(\{v\})=B_\ang(\{\pi^*\})=\{\{\overline{Y}\}\}$
and
$B_\dae(\{\epsilon\})=B_\dae(\{\overline{Y}\})=\{\{\pi^*\}\}$.
Observe that the game has a unique pure Nash equilibrium. Thus, $\PNE=\{(\{\overline{Y}\},\{\pi^*\})\}$.

For $\mathcal U_2 = \langle \mathcal E,\pertur,  \{\epsilon, \overline{Y}\}, 
\{v,\pi^*\}, 1,1, Y, -\pi\rangle$ the $\ang/\dae$ game $\Gamma (\mathcal U_2)$ is described as follows.

\medskip
{\renewcommand{\arraystretch}{1.2}
\begin{tabular}{cc|c|c|}
\cline{3-4}
& &\multicolumn{2}{|c| }{$\dae$}\\
\cline{3-4}
   & & $\{v\}$    & $\{\pi^*\}$\\
\cline{1-4}
\multicolumn{1}{|c|}{\multirow{2}{*} {$\ang$}} &
\multicolumn{1}{|c|}{$\{\epsilon\}$}   
& $1314/13, -62/13$
& $1322/13,-77/13$\\
\cline{3-4}
\multicolumn{1}{|c|}{}   &
\multicolumn{1}{|c|}{$\{\overline{Y}\}$}   
& $1623/13, -58/13$
& $1631/13,-73/13$\\
\cline{1-4}
\end{tabular}
}

\bigskip

The sets of $\dae$\rq{}s  best responses are
$B_\dae(\{\epsilon\})=B_\dae(\{\overline{Y}\})=\{\{v\}\}$. So, 
$\PNE=\{(\{\overline{Y}\},\{v\})\}$.
\hfill
\qed
\end{ex}

%
In view of the previous examples an important question is whether we can characterize uncertainty profiles with respect to the existence of \PNE in the associated \angdae game. In the following section we provide such an answer for particular cases generalizing the situations presented in the previous examples. Now we turn our attention to the more general notion of \emph{mixed Nash equilibrium}.

\remove{having no, u \PNE or a unique \Pthe  responses appearing in Example 
\ref{FirstGame} can be  generalised to any numerical example. This is not straightforward because, as  $b\in \mathcal A$ the value of $g=(1-b)f$ is modified in $\pertur$ into a new $g'$ and the factorization 
\[
\left( 
\begin{array}{c}
Y(a,d)\\
r(a,d)
\end{array}
\right)=
\left( 
\begin{array}{c}
Y\\
r
\end{array}
\right)+ \cdots
\]
is not easy because the $g'$ appears as a quotient of both, $Y(a,d)$ and
$r(a,d)$ and the linear structure is broken. When $\mathcal A=\mathcal D=\{G,T\}$, Lemma \ref{EasyFiscalPolicies}  applies and we get a general result in Theorem \ref{ThmFiscalPolicy}.
Taking into account Lemma \ref{JustPerturbationsIS-MP}, part of the results given in Exercise \ref{IS-MP-game} are generalized in Theorem \ref{IS-MP-Thm1}.}

It is well know that although a strategic game might not have a PNE, it always has a  Nash equilibrium on mixed strategies \citep{OsRu94}.  So,  from the  \angdae game, the stable perturbed situations are described by the mixed strategies of the players in the  Nash equilibria. 
A \emph{mixed} strategies for a  player is a probability distribution on its set of actions. Thus in an \angdae game, 
a \emph{mixed strategy profile} is a tuple $(\alpha,\beta)$ where $\alpha: A_\ang\to [0,1]$ and $\beta: A_\dae\to [0,1]$ are probability distributions.
The utility for player $\agent\in\{\ang,\dae\}$  of a mixed strategy profile $(\alpha,\beta)$ is defined as
\[u_\agent(\alpha,\beta) = \sum_{(a,d)\in A_\ang\times A_\dae}\alpha(a)\beta(d) u_\agent(a,d)\]
The following example illustrates the notion of mixed strategy.
\begin{ex} We provide an example of mixed strategies, for the game  $\Gamma(\mathcal U_1)$ defined in
Example \ref{FirstGame}. Let $(\alpha,\beta)$ a mixed strategy for $\Gamma(\mathcal U_1)$. 

As $A_\ang = \{\{b\},\{G\}\}$ we can describe  $\alpha$ as a function that assigns probability $x$ to $\{b\}$ and probability $1-x$ to $\{G\}$, for some $0\leq x\leq 1$. We represent such a function as 
\[\alpha =\big(\alpha(\{b\}),\alpha(\{G\})\big) = (x, 1-x) \hbox{\;with\;} 0\le x\le 1\] 
As $A_\dae = \{\{P,G\},\{P,T\},\{T,G\}\}$ we can describe a probability distribution on  $A_\dae$ as
\begin{eqnarray*}
&&\beta=\big(\beta(\{P,G\}),\beta(\{P,T\}),\beta(\{T,G\})\big)\\
&&\qquad\qquad=(y, z, 1-y-z) \hbox{\;with\;} 0\le y,z\le 1  \hbox{\; and\;} y+z\le 1
\end{eqnarray*}
We can represent a generic mixed strategy profile using the bi-matrix form of $\Gamma(\mathcal U_1)$ and adding  the probabilities of each action.

\bigskip
{\renewcommand{\arraystretch}{1.2}
\begin{tabular}{cc|c|c|c|}
\cline{3-5}
& &\multicolumn{3}{|c| }{$\dae$}\\
\cline{3-5}
   & & $\{P,G\}$    & $\{P,T\}$ & $\{T,G\}$\\
 & & $y$    & $z$ & $1-y-z$\\
\cline{1-5}
\multicolumn{1}{|c|}{\multirow{2}{*} {$\ang$}} &
\multicolumn{1}{|c|}{$\{b\}$\qquad $x$}   
& $1062.96, 7.29$
& $1029.62, 6.962$
& $1233.33, 22/3\approx 7.33$\\
\multicolumn{1}{|c|}{}   &
\multicolumn{1}{|c|}{$\{G\}$\quad $1-x$}   
& $1066.66, 22/3\approx 7.33$
& $1041.66, 7.08$
& $1075, 5.75$\\
\cline{1-5} 
\end{tabular}
}

\bigskip
The utility, for $\agent\in\{\ang,\dae\}$, is computed as
\begin{eqnarray*}
u_{\agent}(\alpha,\beta) 
&=&\alpha(\{b\})\beta(\{P,G\})u_{\agent}(\{b\},\{P,G\}) +\cdots\\
&&\cdots +\alpha(\{G\})\beta(\{T,G\})u_{\agent}(\{G\},\{T,G\}) 
\end{eqnarray*}
In particular setting  $x= 1/3$, $y=1/4$ and $z=1/2$. The strategy profile is  
$(\alpha,\beta)=\big((1/3, 2/3), (1/4, 1/2,1/4)\big)$ and the utilitities for the two players are
\begin{eqnarray*}
u_{\ang}(\alpha,\beta) 
&=&\frac{1}{3}\times\frac{1}{4}\times 1062.96 +\cdots=1067.124\\
u_{\dae}(\alpha,\beta) 
&=&\frac{1}{3}\times\frac{1}{4}\times 7.29 +\cdots= 6.9195
\qquad\qquad\qquad\qquad\qquad\qquad\qquad\qquad\qquad\Box 
\end{eqnarray*}
\end{ex}
Let $\Delta_\ang$ and $\Delta_\dae$ denote the set of mixed strategies for  players $\ang$ and $\dae$, respectively. 
A pure  strategy profile $(a, d)$ is a special case of a mixed strategy profile $(\alpha,\beta)$ 
in which $\alpha(a)=1$ and $\beta(d)=1$. 
The definition of Nash equilibrium extends the conditions for pure strategies to mixed strategies.
We adapt to  $\ang/\dae$ games the characterization of a mixed Nash equilibrium given in  \citep{Osborne2004} that we will use to prove that some mixed strategies are Nash equilibria. 

\begin{property}\label{mixedNash}
A mixed strategy profile $(\alpha, \beta)$ is a mixed Nash equilibrium 
in $\Gamma(\mathcal U)$ if  the following two symmetric conditions hold.
For all $a\in A_\ang$, when $\alpha(a)>0$ we have 
$u_{\ang}(\alpha,\beta)= u_{\ang}(a,\beta)$, otherwise  $u_{\ang }(\alpha,\beta)\ge u_{\ang}(a,\beta)$.
For all  $d\in A_\dae$, when $\beta(d)>0$ we have
$u_{\dae}(\alpha,\beta)= u_{\dae}(\alpha,d)$, otherwise $u_{\dae}(\alpha,\beta)\ge u_{\dae}(\alpha,d)$.
\end{property}

In the following example  we use this characterization to obtain a mixed Nash equilibrium.


\begin{ex}
\label{ExampleMixedNash}
Let us continue with the Example \ref{FirstGame}.  We have seen that the $\ang/\dae$ game corresponding to 
$\mathcal U_2 = \langle \mathcal E,  \pertur,\{b,G\}, \{P,G, T\}, 1, 2, Y, -r\rangle$
has no \PNE. Let us compute a mixed Nash equilibrium. 
An strategy $\alpha\in \Delta_\ang$ can be described as
$\alpha = (\alpha(\{b\}), \alpha(\{G\}))=(x, 1-x)$ with $0\le x\le 1$.
As $\dae$ is rational,  $\dae$ will never choose 
$\{P,G\}$ because, independently of  $\ang$\rq{}s action,  $\{P,T\}$ or $(\{T,G\})$ provide better utility. Therefore, according to Property \ref{mixedNash}, $\beta(\{P,G\})=0$ in any  Nash equilibria $(\alpha,\beta)$. Therefore, we can restrict our search to  $\beta\in \Delta_\dae$ having the form
\[\beta=(\beta(\{P,G\}), \beta(\{P,T\}), \beta(\{T,G\}))=(0, y, 1-y),\] 
for  $0\le y\le 1$.  Using property \ref{mixedNash} we can express the $\dae$ utilities and the conditions for such an  $(\alpha,\beta)$ to be a mixed Nash equilibrium.
\begin{align*}
u_{\dae}(\alpha, \{P,G\}) &
=- x\times 7.29 - (1-x)\times\frac{22}{3}
= 0.04333333\times x - 7.333333\\
u_{\dae}(\alpha, \{P,T\}) &
=- x\times 6.962  - (1-x)\times 7.08
=0.118\times x- 7.08\\
u_{\dae}(\alpha, \{T,G\}) &
=-x\times\frac{22}{3} - (1-x)\times 5.75
= - 1.583333\times x - 5.75\\
u_{\dae}(\alpha, \{P,G\}) &\le 
u_{\dae}(\alpha, \{P,T\})=
u_{\dae}(\alpha, \{T,G\})
\end{align*}
The equality  $u_{\dae}(\alpha, \{P,T\})=
u_{\dae}(\alpha, \{T,G\})$ determines  $\alpha =(0.78174, 0.21826)$

To obtain  the value of $y$ we state the utilities and conditions for $\ang$.
\begin{align*}
u_{\ang}(\{b\},\beta) &
= y\times  1029.62 +(1-y)\times  1233.33
= -203.71\times y +  1233.33\\
u_{\ang}(\{G\},\beta) &= 
y\times 1041.66 +(1-y)\times  1075 =   -33.34\times y+ 1075
\end{align*}
To get a Nash equilibrium we need $u_{\ang}(\{b\},\beta) = u_{\ang}(\{G\},\beta)$ then 
\[\beta =(0, 0.9293303,  0.07066972).\]

%

Therefore,   
$(\alpha, \beta)=( (0.78174, 0.21826)  , (0, 0.9293303,  0.07066972)$ is a mixed Nash equilibrium for $\Gamma (\mathcal U_2)$.
Finally, we can compute the players\rq{} utilities in such an equilibrium:
\begin{align*}
u_{\ang}(\alpha,\beta)  & = x\times y\times 1029.62+x\times (1-y)\times   
 1233.33 +(1-x)\times y \times 1041.66\\
& \quad + (1-x)\times(1-y)\times 1075 =1044.016\\
u_{\dae}(\alpha,\beta) =&- x*y*6.962-x*(1-y)*(22/3)
-(1-x)*y*7.08\\
& \quad - (1-x)*(1-y)*5.75=
-6.98775 
\end{align*}
\vskip -20pt \ \hfill $\Box$
\end{ex}
As we have seen in the previous examples games can have more than one Nash equilibrium and  get different utilities in those situations. 
This property does not happen in the case of two players \emph{zero-sum games} \citep{vonNeumannMorgenstern}.  In a zero-sum game, the sum of the utilities of both players is always zero. A zero-sum game  can have several Nash equilibria but all of them provide the same utility for the players. As all the Nash equilibria have the same utility for the first player, this utility is called the \emph{value of the game}. 

To obtain  a zero-sum \angdae game, we need   uncertainty profiles where $u_{\dae} + u_{\ang} = 0$. We can model this situation by considering a unique objective function $u$ so that $u_{\ang} = u$ and $u_{\dae} = -u$. In such a case the angelic and daemonic interests go exactly in opposite directions.
In a zero-sum situation we simplify notation. In an uncertainty profile we introduce only the function $u$:  
$\mathcal U =\langle \mathcal E, \pertur,\mathcal A, \mathcal D, b_{\ang}, b_{\dae}, u\rangle$.
The description of the corresponding zero-sum \angdae is done by providing  the table with the $u$ values, replacing the usual bi-matrix form. 
The value of the game $\Gamma(\mathcal U)$ is denoted by  $\nu(\mathcal U)$ and it can be expressed as
\[
\nu(\mathcal U)=
\min_{\alpha \in\Delta\ang}\max_{\beta\in\Delta_\dae}
u(\alpha, \beta) =
\max_{\beta\in\Delta_\dae}\min_{\alpha\in\Delta_\ang}u(\alpha, \beta)
\]

One natural way to get a zero-sum situation associated to a given  uncertainty profile  
$\mathcal U =\langle \mathcal E, \pertur,\mathcal A, \mathcal D, b_{\ang}, b_{\dae}, u_{\ang},u_{\dae}\rangle$, is to consider a strategy profile  $(a,d)$
as a perturbation of the strategy profile $(\emptyset,\emptyset)$.  One way of doing so is to consider the uncertainty profile  
$\mathcal U\rq{} =\langle \mathcal E, \pertur,\mathcal A, \mathcal D, b_{\ang}, b_{\dae}, u\rangle$ where 
\[u(a,d)= u_{\ang}(a,d) -\frac{u_{\ang}(\emptyset,\emptyset)}{u_{\dae}(\emptyset,\emptyset)}u_{\dae}(a,d). \]

In the following examples we  analyze the values of such games derived  from  uncertainty profiles on the IS-LM and the IS-MP models considered before.

\begin{ex}
Consider the uncertainty profile  $\mathcal U=\langle \mathcal E,\pertur, \{b,G\}, \{P,G, T\}, 1,2, u\rangle$ with $u(a,d)=Y(a,d)-\frac{Y(\emptyset,\emptyset)}{r(\emptyset,\emptyset)}r(a,d)$ for the IS-LM model derived from the uncertainty profile considered in Example \ref{Mankiw}.
Recall that the system\rq{}s solution is $(Y,r)=(1100,6)$. According to Lemma \ref{lemmaTrivial}, when 
$(a,d)=(\emptyset,\emptyset)$ it holds $(Y(\emptyset,\emptyset),r(\emptyset,\emptyset))=(Y,r)=(1100,6)$
therefore $u(a,d)=Y(a,d)-\frac{1100}{6} r(a,d)$. Simplifying $u(a,d)= Y(a,d)-\frac{550}{3}r(a,d)$.
The  corresponding zero-sum \angdae game $\Gamma(\mathcal U)$ is described by the table of $u$: 
\remove{\bigskip
{\renewcommand{\arraystretch}{1.2}
\begin{tabular}{cc|c|c|c|}
\cline{3-5}
& &\multicolumn{3}{|c| }{$\dae$}\\
\cline{3-5}
   & & $\{P,G\}$    & $\{P,T\}$ & $\{T,G\}$\\
\cline{1-5}
\multicolumn{1}{|c|}{\multirow{2}{*} {$\ang$}} &
\multicolumn{1}{|c|}{$\{b\}$}   
& $u(\{b\},\{P,G\})$, $-u(\{b\},\{P,G\})$
& $\cdots$
&$\cdots$\\
\multicolumn{1}{|c|}{}   &
\multicolumn{1}{|c|}{$\{G\}$}   
& $\cdots$
& $\cdots$
& $\cdots$\\
\cline{1-5} 
\end{tabular}
}

\bigskip
\noindent
where 
$u_{\ang}(a,d)= u(a,d)$ and $u_{\dae}(a,d)= -u(a,d)$. In order to avoid redundant information, as usual in zero-sum games, in the following table  
we give only angel's utility, $u_{\ang}= u$. As daemon's utility is just $-u$
we do not write it explicitly summarizing into

\bigskip
{\renewcommand{\arraystretch}{1.2}
\begin{tabular}{cc|c|c|c|}
\cline{3-5}
& &\multicolumn{3}{|c| }{$\dae$}\\
\cline{3-5}
   & & $\{P,G\}$    & $\{P,T\}$ & $\{T,G\}$\\
\cline{1-5}
\multicolumn{1}{|c|}{\multirow{2}{*} {$\ang$}} &
\multicolumn{1}{|c|}{$\{b\}$}   
& $u(\{b\},\{P,G\})$
& $\cdots$
&$\cdots$\\
\multicolumn{1}{|c|}{}   &
\multicolumn{1}{|c|}{$\{G\}$}   
& $\cdots$
& $\cdots$
& $\cdots$\\
\cline{1-5} 
\end{tabular}
}

\bigskip
\noindent
and computing explicitly the utilities we obtain $\Gamma(\mathcal U)$:}

\bigskip
\begin{footnotesize}
{\renewcommand{\arraystretch}{1.2}
\begin{tabular}{cc|c|c|c|}
\cline{3-5}
& &\multicolumn{3}{|c| }{$\dae$}\\
\cline{3-5}
   & & $\{P,G\}$    & $\{P,T\}$ & $\{T,G\}$\\
\cline{1-5}
\multicolumn{1}{|c|}{\multirow{2}{*} {$\ang$}} &
\multicolumn{1}{|c|}{$\{b\}$} 
&$-22250/81 \approx-274.69$   
&$-20000/81\approx -246.91$  
&$-1000/9\approx -111.11$\\
\multicolumn{1}{|c|}{}   &
\multicolumn{1}{|c|}{$\{G\}$}  
&$-2500/9\approx -277.77$   
&$-4625/18\approx -256.94$ 
& $125/6\approx 20.83$\\
\cline{1-5} 
\end{tabular}
}
\end{footnotesize}

\bigskip
Observe that there is one \PNE at $(\{b\},\{P,G\})$, therefore we know that all \NE will provide the same utility and that  $\nu(\mathcal U)=-22250/81$.
\hfill $\Box$
\end{ex}

\begin{ex} Let us reconsider the 
Example \ref{IS-MP-game}.
This leads to the zero-sum e uncertainty profile 
$\mathcal U = \langle \mathcal E,\pertur,  \{\epsilon, \overline{Y}\}, 
\mathcal D=\{v,\pi^*\}, 1,1, u\rangle$ where  $u(a,d)=Y(a,d)-(Y(\emptyset,\emptyset)/\pi(\emptyset,\emptyset))\pi(a,d)$. According to 
Example \ref{first-IS-MP},  $Y(\emptyset,\emptyset)= 1306/13$ and $\pi(\emptyset,\emptyset)=34/13$. In this case $\Gamma(\mathcal U)$ is described by the following table.

\medskip
{\renewcommand{\arraystretch}{1.2}
\begin{tabular}{cc|c|c|}
\cline{3-4}
& &\multicolumn{2}{|c| }{$\dae$}\\
\cline{3-4}
   & & $\{v\}$    & $\{\pi^*\}$\\
\cline{1-4}
\multicolumn{1}{|c|}{\multirow{2}{*} {$\ang$}} &
\multicolumn{1}{|c|}{$\{\epsilon\}$}   
& $-1396/17\approx -82.11$
& $-2139/17\approx -125.82$\\
\cline{3-4}
\multicolumn{1}{|c|}{}   &
\multicolumn{1}{|c|}{$\{\overline{Y}\}$}   
& $-791/17\approx -46.52$
& $-1534/17\approx -90.23$\\
\cline{1-4}
\end{tabular}
}

\bigskip
The sets of best responses are
$B_\ang(\{v\})=B_\ang(\{\pi*\})=\{\{\overline{Y}\}\}$
and
$B_\dae(\{\epsilon\})=B_\dae(\{\overline{Y}\})=\{\{\pi^*\}\}$.
There is a  \PNE   $(\{\overline{Y}\},\{\pi^*\})$ and therefore 
the value of the game is $\nu(\mathcal U)= -1534/17\approx -90.23$.
\hfill\qed
\end{ex}
%


\section{Uncertainty in the IS-LM Model}
\label{section-IS-LM}
We consider the  case in which a perturbation  is exerted only on  the  fiscal policy, i.e., $\mathcal A=\mathcal D=\{G,T\}$. Furthermore, we assume also that $\ang$ and $\dae$ can control just one of the components.
We analyze two situations with respect to the utilities. In the first case $\ang$'s objective is  to increase the income as much as possible while $\dae$ tries to increase the interest rate.    In the second case we considerer a zero-sum approach  taking 
$u=Y-k \, r$.

\subsection{A Case of Fiscal Policy under Uncertainty in the IS-LM Model}
\label{Fiscal policy}
Consider the case where  $\ang$ and $\dae$ have the capability to act over $T$ and $G$, that is  $A_\ang=A_\dae=\{\{T\},\{G\}\}$ and 
$u_\ang = Y$ and $u_\dae=r$ (as it was the case in Example \ref{FirstGame}).
We ask if the addition of  uncertainty can generate a  situation in which no \PNE exists. We  provide a negative answer.  Before stating the result. 
Let us remind, in the context of \angdae games, the definition of {\em dominant strategy equilibrium} (\DSE)  taken from \citep{OsRu94}.
A  \DSE is a joint action  $(a,d)$ such that, for any other joint action $(a',d')\in A_\ang\times A_\dae$,  
$u_\ang(a,d')\ge u_\ang(a',d')$ 
and 
$u_\dae(a',d)\ge u_\dae(a',d')$. 
As pointed  out in \citep{OsRu94},  in a \DSE,  
the action of \emph{every} player is a best response independently on which are the actions taken by the other players. 
%
Our next result proves the existence of \DSE when the perturbation has to be exerted only in one component
\begin{theorem}
\label{ThmFiscalPolicy}  Let $\pertur$ be a perturbation strength model, for the IS-LM model, given by 
\[
{\renewcommand{\arraystretch}{1.2}
\begin{tabular}{|| c || c |c|c|c||}
\hline
agent &$a, b, c, d, e, f$ & $T$ &  $G$ & $M, P$\\
\hline
$\ang$ &$0$&  $\delta_\ang(T)$  & $\delta_\ang(G)$ &$0$\\
$\dae$ &$0$ &  $\delta_\dae(T)$ & $\delta_\dae(G)$  & $0$\\
\hline
\end{tabular}
}
\]  
For the uncertainty profile $\mathcal U=\langle \mathcal E,   \pertur, \{G,T\}, \{G,T\}, 1, 1, Y, r\rangle$,  the associated \angdae game $\Gamma(\mathcal U)$ in bi-matrix form is given by

\bigskip
{\renewcommand{\arraystretch}{1.5}
\begin{tabular}{cc|c|c|c}
\cline{3-4}
& &\multicolumn{2}{|c| }{$\dae$}\\
\cline{3-4}
   & & $\{T\}$    & $\{G\}$ \\
\cline{1-4}
\multicolumn{1}{|p{0.02cm}|}{\multirow{4}{*} {$\ang$}} 
&\multicolumn{1}{|p{0.45cm}|}{\multirow{2}{*} {$\{T\}$}} 
 &\rule{0pt}{11pt}$u_\ang=Y-\frac{f}{g}b(\delta_\ang(T)+\delta_\dae(T))$ 
& $u_\ang=Y+\frac{f}{g}(\delta_\dae(G)-b\delta_\ang(T))$\\
\multicolumn{1}{|c|}{}   
& &$u_\dae=r-\frac{e}{g}b(\delta_\ang(T)+\delta_\dae(T))$
& $ u_\dae=r+\frac{e}{g}(\delta_\dae(G)-b\delta_\ang(T))$\\   
\cline{2-4}
\multicolumn{1}{|p{0.02cm}|}{}  
&\multicolumn{1}{|p{0.01cm}|}{\multirow{2}{*} {$\{G\}$}} 
&\rule{0pt}{11pt}$u_\ang=Y+ \frac{f}{g}(\delta_\ang(G)-b\delta_\dae(T))$ & 
$u_\ang=Y+ \frac{f}{g}(\delta_\ang(G)+\delta_\dae(G))$ \\
\multicolumn{1}{|c|}{}   
&&$u_\dae= r+\frac{e}{g}(\delta_\ang(G)-b\delta_\dae(T))$
&$u_\dae= r+\frac{e}{g}(\delta_\ang(G)+\delta_\dae(G))$\\
\cline{1-4}
\end{tabular}
}

\bigskip
Furthermore, it holds that $\Gamma(\mathcal U)$ has always  a \DSE.  
\end{theorem}

\begin{proof}
Observe that, in $\pertur$, $\delta_\ang(x)=\delta_\dae(x)=0$, for $x\notin\{G,T\}$. 
Furthermore, in  $\Gamma(\mathcal U)$ as $\mathcal A= \{T,G\}$ but $b_{\ang}=1$, the angel $\ang$ has only the possibility to act over one parameter  therefore $A_\ang =\{\{T\},\{G\}\}$ and similarly for $\dae$. 
Thus, in $\Gamma(\mathcal U)$  the sets of actions are 
$A_\ang =A_\dae=\{\{T\},\{G\}\}$. As $u_\ang(a,d)=Y(a,d)$ and $u_\dae(a,d)=r(a,d)$ the $\ang/\dae$ game is

\bigskip
{\renewcommand{\arraystretch}{1.2}
\begin{tabular}{cc|c|c|}
\cline{3-4}
& &\multicolumn{2}{|c| }{$\dae$}\\
\cline{3-4}
   & & $\{T\}$    & $\{G\}$\\
\cline{1-4}
\multicolumn{1}{|c|}{\multirow{2}{*} {$\ang$}} &
\multicolumn{1}{|c|}{$\{T\}$}   
& $Y(\{T\},\{T\}), r(\{T\},\{T\})$
& $Y(\{T\},\{G\}), r(\{T\},\{G\})$\\
\cline{3-4}
\multicolumn{1}{|c|}{}   &
\multicolumn{1}{|c|}{$\{G\}$}   
& $Y(\{G\},\{T\}), r(\{G\},\{T\})$
& $Y(\{G\},\{G\}), r(\{G\},\{G\})$\\
\cline{1-4}
\end{tabular}
}

\bigskip
According to Lemma \ref{EasyFiscalPolicies}
\[Y(a,d)=Y+\frac{f}{g}\delta_{\pertur}(G,T)[a,d],\;
r(a,d)=r+\frac{e}{g}\delta_{\pertur}(G,T)[a,d].
\]
To find the values $Y(a,d)$ and $r(a,d)$ 
we need  to  compute  the values of 
$\delta_\pertur(G,T)[a,d]= \delta_\pertur(G)[a,d]-b\delta_\pertur(T)[a,d]$. 
For instance, when  $(a,d)=(\{T\},\{T\})$, according to Definition \ref{DefDeltaStress}, 
$\delta_{\pertur}(G)[\{T\},\{T\}]=0$,
$\delta_{\pertur}(T)[\{T\},\{T\}]=\delta_\ang(T)+\delta_\dae(T)$
and 
$\delta_\pertur(G,T)[\{T\},\{T\}]= -b(\delta_\ang(T)+\delta_\dae(T))$.
The other  cases are
$\delta_\pertur(G,T)[\{T\},\{G\}]=\delta_\dae(G)-b\delta_\ang(T)$,
$\delta_\pertur(G,T)[\{G\},\{T\}]=\delta_\ang(G)-b\delta_\dae(T)$  and
$\delta_\pertur(G,T)[\{G\},\{G\}]=\delta_\ang(G)+\delta_\dae(G)$.
From those computations we conclude that the bi-matrix  form of the $\ang/\dae$ corresponds to the statement.

In order to analyse  the structure of Nash equilibria,
let us study best responses by case analysis.  Consider $\ang$'s utility increment  when $\dae$ is in $\{T\}$ and $\ang$ moves from $\{T\}$ to $\{G\}$ written as $T\to G$.
\begin{eqnarray*}
&&u_{\ang}(\{G\},\{T\})-u_{\ang}(\{T\},\{T\}) = Y(\{G\},\{T\})-Y(\{T\},\{T\}) \\
&&\quad=(f/g)\big(\delta_\pertur (G,T)[\{G\},\{T\}]-\delta_\pertur(G,T)[\{T\},\{T\}]\big)\\
&&\quad=(f/g)(b\delta_\ang(T)+\delta_\ang(G))
\end{eqnarray*}
Defining $\mu_{\ang,T\to G}=b\delta_\ang(T)+\delta_\ang(G))$ we write 
$u_{\ang}(\{G\},\{T\})-u_{\ang}(\{T\},\{T\}) = (f/g)\mu_{\ang,T\to G}$. We see that, abstracting from factor $f/g$,   the expression $\mu_{\ang,T\to G}$ measures the increment (in particular the sign) of $\ang$'s utility. Due to the linear structure of the model, 
 $\mu_{\dae,T\to G}=b\delta_\dae(T)+\delta_\dae(G))$ also makes sense to measures  $\dae$'s increment.
We prove the following claim: 
given $\agent\in\{\ang,\dae\}$ and defining
$\mu_{\agent, T\to G}
=b\delta_\agent(T)+\delta_\agent(G)$
it holds that 
\[
B_\agent(\{T\})= B_\agent(\{G\})= 
\begin{cases} 
\{G\}& \text{if $\mu_{\agent, T\to G}>0$}\\
\{T,G\} &\text{if $\mu_{\agent, T\to G}=0$}\\
\{T\}& \text{if $\mu_{\agent, T\to G}<0$}\\
\end{cases}
\]
To prove the claim,  we start with $\ang$
and consider $B_\ang(\{T\})$ and $B_\ang(\{G\})$ separately.
In order to find $B_\ang(\{T\})$, we know the change when $\ang$ moves between $\{T\}$ and $\{G\}$, but
$Y(\{G\},\{T\})-Y(\{T\},\{T\})=
(f/g)\mu_{\ang,T\to G}$.
For   $B_\ang(\{G\})$, we have also
$Y(\{G\},\{G\})-Y(\{T\},\{G\}) =(f/g)\mu_{\ang,T\to G}$ .
Therefore, for $d\in A_\dae$,
we have
$
Y(\{G\},d)-Y(\{T\}, d) 
=(f/g)\mu_{\ang, T\to G}.
$
As $f/g>0$, we conclude that
\[
B_\ang(\{T\})= B_\ang(\{G\})= 
\begin{cases} 
\{G\}& \text{if $\mu_{\ang, T\to G}>0$}\\
\{T,G\} &\text{if $\mu_{\ang, T\to G}=0$}\\
\{T\}& \text{if $\mu_{\ang, T\to G}<0$}\\
\end{cases}
\]
For $\dae$,  we have 
$r(a,\{G\})-r(a,\{T\})=(e/g)\mu_{\dae, T\to G}$ for $a\in A_\ang$.
As $e/g>0$, we obtain
\[
B_\dae(\{T\})= B_\dae(\{G\})= 
\begin{cases} 
\{G\}& \text{if $\mu_{\dae, T\to G}>0$}\\
\{T,G\} &\text{if $\mu_{\dae, T\to G}=0$}\\
\{T\}& \text{if $\mu_{\dae, T\to G}<0$}\\
\end{cases}
\]
As $\ang$ and $\dae$ have similar structure,  taking any agent 
$\agent\in \{\ang, \dae\}$,
the claim follows. 

According to the best responses, by case analysis,  we find the following \PNE structure:

\medskip
{\renewcommand{\arraystretch}{1.2}
\begin{tabular}{ c ||c | c | c }
 & $\mu_{\ang, T\to G}>0$ &$\mu_{\ang, T\to G}=0$ &$\mu_{\ang, T\to G}<0$\\
\hline
$\mu_{\dae, T\to G}>0$ &$\{G,G\}$&$\{G,T\}$,$\{G,G\}$&$\{G,T\}$\\
$\mu_{\dae, T\to G}=0$ &$\{T,G\}$
$\{G,G\}$
&$\{G,G\}$$\{G,T\}$$\{T,G\}$$\{T,T\}$&$\{T,T\}$$\{G,T\}$\\
$\mu_{\dae, T\to G}<0$ &$\{T,G\}$&$\{T,T\}$$\{T,G\}$&$\{T,T\}$
\end{tabular}
}

\medskip
To prove the existence of a  \DSE we consider first the  case 
$\mu_{\agent, T\to G}>0$, for $\agent\in\{\ang,\dae\}$. Let us show that, in such a  case, 
$(\{G\},\{G\})$ is a  dominant strategy equilibrium. 
When $\agent=\ang$ the conditions for \DSE are 
$u_\ang(\{G\}, d)\ge u_\ang(\{T\}, d)$ for $d\in A_\dae$
As for $d\in A_\dae$
\begin{eqnarray*}
u_\ang(\{G\},d)-u_\ang(\{T\},d) =
(f/g)\mu_{\ang, T\to G}>0
\end{eqnarray*}
the constraint hold.
When $\agent=\dae$ the condition  is
$u_\dae(a,\{G\})\ge u_\dae(a,\{T\})$ 
for any $a\in A_\ang$. As
\[u_\dae(a,\{G\})-u_\dae(a,\{T\})= (e/g)\mu_{\dae, T\to G}>0
\]
constraints hold and 
$(\{G\},\{G\})$ is a  dominant strategy equilibrium  when 
$\mu_{\ang, T\to G}>0$. The remaining cases follow by a similar argument.
\end{proof}

In the situation analyzed on Theorem \ref{ThmFiscalPolicy} there is the possibility of  modifying government spending or taxes. For both components  there are positive and negative aspects, this is model by the perturbation strength  parameters
$\delta_\ang(G)$, $\delta_\dae(G)$, $\delta_\ang(T)$, and $\delta_\dae(T)$.
If we take into account {\em all} the possible perturbed  actions  together we get that
 \[\hat G= G+ \delta_\ang(G)+ \delta_\dae(G),\;
\hat T= T+\delta_\ang(T)+ \delta_\dae(T).\]
In such a case, the income $\hat Y$ and the interest rate $\hat r$ are given by
\[
\left(
\begin{array}{c}
\hat Y\\
\hat r
\end{array}
\right)
=
\frac{1}{g}\left(
\begin{array}{cc}
f & d/P\\
e &\hbox{\;\;}-(1-b)/P
\end{array}
\right)
\left(
\begin{array}{c}
a+c+\hat G-b\, \hat T\\
M
\end{array}
\right), 
\]
by linearity, we have  
\[
\left(
\begin{array}{c}
\hat Y\\
\hat r
\end{array}
\right)
=
\left(
\begin{array}{c}
Y\\
r
\end{array}
\right)
+\big(
\delta_\ang(G)+\delta_\dae(G)-b \delta_\ang(T)-b\delta_\dae(T)\big)
\frac{1}{g}\left(
\begin{array}{c}
f \\
e 
\end{array}
\right).
\]
The use of Game Theory gives the opportunity to study this situation under different "intermediate" scenarios with a  larger number of possibilities to consider. 
We adopted in Theorem \ref{ThmFiscalPolicy} a moderate view respecting to goodness and badness: there is too much luck in assuming that measures over $G$ and $T$ will be successful all together. This translates in the rule:  it is possible to apply $\delta_\ang(G)$ or $\delta_\ang(T)$ but not both {\em at the same time}. 
Our point about badness is similar, things are bad but not too bad, translating into similar rules for $\dae$'s actions.  
We are {not} saying that it is not possible to simultaneously change  government spending and taxation. But if this is the case, the design of the uncertainty profile, and in consequence of the \angdae game, will be different. 
\remove{
This can happen when $\ang$ can act over taxation $\{T\}$
and $\dae$ can act over government spending $\{G\}$. To be more precise, when we say that 
it is not possible to apply $\delta_\ang(G)$ and $\delta_\ang(T)$ {\em at the same time}, we  say that in the $\ang$-$\dae$ game corresponding to Theorem \ref{ThmFiscalPolicy} a strategy like $(\{G,T\},\dots)$ does not appear.}

Let us comment Theorem \ref{ThmFiscalPolicy}.
The possible results depends on $\mu_{\agent,T\to G}$ for $\agent\in\{\ang,\dae\}$. Let us consider some cases. 
When $\mu_{\agent,T\to G}>0$ for $\agent\in\{\ang,\dae\}$ the strategy profile $(\{G\}\{G\})$ is the unique \PNE with utilities
\[Y(\{G\}\{G\})=  Y+ \frac{f}{g}(\delta_\ang(G)+\delta_\dae(G)),
r(\{G\}\{G\})=  Y+ \frac{e}{g}(\delta_\ang(G)+\delta_\dae(G)).\]
In such a case, just government spending is a good policy rule. Note that 
perturbations exerted  to taxes ($\delta_\ang(T)$, $\delta_\dae(T)$ do not appear). Therefore, it seems better to keep taxes unchanged.  The relation between both approaches is given by 
\[
\left(
\begin{array}{c}
\hat Y\\
\hat r
\end{array}
\right)
=
\left(
\begin{array}{c}
Y(\{G\}\{G\})\\
r(\{G\}\{G\})
\end{array}
\right)
-b\big( \delta_\ang(T)+\delta_\dae(T)\big)
\frac{1}{g}\left(
\begin{array}{c}
f \\
e 
\end{array}
\right).
\]
Thus, depending on the sign of $\delta_\ang(T)+\delta_\dae(T)$, one solution or the other one will be the best one. 
When $\mu_{\agent,T\to G}< 0$, for $\agent\in\{\ang,\dae\}$, the strategy profile 
$(\{T\}\{T\})$ is the unique \PNE. In such a case taxing is the good solution and the relation between both approaches is given by
\[
\left(
\begin{array}{c}
\hat Y\\
\hat r
\end{array}
\right)
=
\left(
\begin{array}{c}
Y(\{T\}\{T\})\\
r(\{T\}\{T\})
\end{array}
\right)
+\big(\delta_\ang(G)+\delta_\dae(G)\big)
\frac{1}{g}\left(
\begin{array}{c}
f \\
e 
\end{array}
\right).
\]
An interesting case happens when  $\mu_{\ang,T\to G}>0$ and  $\mu_{\dae,T\to G}<0$. In this situation $\{T,G\}$ is the unique \PNE with values
\[Y(\{T\}\{G\})=  Y+ \frac{f}{g}(\delta_\dae(G)-b \delta_\ang(T)),\;
r(\{T\}\{G\})=  Y+ \frac{e}{g}(\delta_\dae(G)-b\delta_\ang(T)).\]
Thus highlighting   a case where both, government spending and taxing, will occur in equilibrium. Comparing with $\hat Y$ and 
$\hat r$ we have that
\[
\left(
\begin{array}{c}
\hat Y\\
\hat r
\end{array}
\right)
=
\left(
\begin{array}{c}
Y(\{T\}\{T\})\\
r(\{T\}\{T\})
\end{array}
\right)
+\big(\delta_\ang(G)-b\delta_\dae(T)\big)
\frac{1}{g}\left(
\begin{array}{c}
f \\
e 
\end{array}
\right).
\]

Uncertainty profiles help us to catch and deal with many intermediate cases between total successfulness and complete failure. Note that uncertainty profiles are flexible tools. For instance we can catch the solution $\hat Y$ and $\hat r$ with an uncertainty profile as we show in the following example

\begin{ex}
Take a perturbation strength model $\pertur$ as in Theorem \ref{ThmFiscalPolicy} and consider the uncertainty profile 
$\mathcal U = \langle \mathcal E, \pertur, \{T,G\}, \{T,G\},  2,2, Y,r\rangle$.
Lemma \ref{LemmaNoFreedom} deals with this case having 
$(\{T,G\},\{T,G\})$ as a unique \PNE. The corresponding utilities are $\hat Y$ and $\hat r$. 
\hfill
\qed 
\end{ex}


\subsection{On Balanced Budgets}
\label{balanced}
Now we consider an $\ang/\dae$ approach to deal with balanced budgets. To do that, we take a  special case of the perturbation strength model $\pertur$ considered in Theorem \ref{ThmFiscalPolicy} where the government spends the amount $\delta$ collected by taxation. We call it $\pertur_{\mathit{balanced}}$, 
\[
{\renewcommand{\arraystretch}{1.2}
\begin{tabular}{|| c || c |c|c|c||}
\hline
agent &$a, b, c, d, e, f$ & $T$ &  $G$ & $M, P$\\
\hline
$\ang$ &$0$&  $0$  & $\delta$ &$0$\\
$\dae$ &$0$ &  $\delta$ & $0$  & $0$\\
\hline
\end{tabular}
}
\]  
Consider the 
uncertainty profile $\mathcal U_{\mathit{balanced}}=\langle \mathcal E,   \pertur_{\mathit{balanced}}, \{G,T\}, \{G,T\}, 1, 1,Y, -r\rangle$.  From Theorem \ref{ThmFiscalPolicy},  we easily get the following description for $\Gamma(\mathcal U_{\mathit{balanced}})$ 
\begin{center}
{\renewcommand{\arraystretch}{1.5}
\begin{tabular}{cc|c|c|c}
\cline{3-4}
& &\multicolumn{2}{|c| }{$\dae$}\\
\cline{3-4}
   & & $\{T\}$    & $\{G\}$ \\
\cline{1-4}
\multicolumn{1}{|p{0.02cm}|}{\multirow{2}{*} {$\ang$}} 
& 
 {$\{T\}$}
 &\rule{0pt}{11pt}$ Y-\frac{f}{g}b\delta$, $-\big(r-\frac{e}{g}b\delta\big)$
 & $Y$, $-r$\\
%
\cline{2-4}
\multicolumn{1}{|p{0.02cm}|}{}  
& $\{G\}$
&\rule{0pt}{11pt}$Y+ \frac{f}{g}\delta(1-b)$, $-\big( r+\frac{e}{g}\delta(1-\delta)\big)$ & 
$Y+ \frac{f}{g}\delta$ , $-\big(r+\frac{e}{g}\delta\big)$\\

\cline{1-4}
\end{tabular}
}
\end{center}

The sets of best responses are
$B_\ang(\{T\})=B_\ang(\{G\})=\{G\}$
and
$B_\dae(\{T\})=B_\dae(\{G\})=\{T\}$.
The the only \PNE is $(\{G\},\{T\})$. This Nash equilibrium corresponds to the case where Government spends the amount $\delta$ collected by taxes and this corresponds to a balanced budget.
%

\subsection{Direct Control over Income and Interest Rate from $\ang$ and $\dae$}
In various  previous examples giving raise to non zero-sum games we have considered   $u_\ang = Y$ and 
$u_\dae=r$. In such cases,  the goal of  $\ang$ consists on maximize the $Y$ (a good thing) while $\dae$  tries to maximize the interest (a bad thing). 
We consider now the case where $\ang$ is interested in controlling  both, $Y$ and $r$,  directly. Informally  $\ang$  likes to keep $Y$ as high as possible and in order to get that  $r$ needs to be as small as possible.
A way to model this situation is to consider $u_\ang = Y-r$. Now $\ang$ is interested in increasing $Y$ and decreasing $r$ in order to keep $u_\ang$ as high as possible. We can tune the relevance of $Y$ in front of $r$ using a parameter  $k>0$ and defining 
$u_\ang(a,d)= Y(a,d)-k\, r(a,d)$. 
In the case that  $\dae$ is interested just in the opposite,  $u_\dae(a,d) = k\, r(a,d)-Y(a,d)$. Therefore, $u_\ang = -u_\dae$ and  $\Gamma(\mathcal U)$ is a zero sum game with $u(a,d)=u_\ang(a,d)$.
Our next results analyzes the \PNE of a case of perturbed fiscal policy where any effort undertaken by $\ang$ can be cancelled by $\dae$ and vice-versa,  this can be modeled by letting 
$\delta_\ang(T) + \delta_\dae(T)=0$ and $\delta_\ang(G) + \delta_\dae(G)=0$. 
%
\begin{theorem}
\label{k-opposite}
Let $\pertur$ be a perturbation strength model for the IS-LM model given by \[
{\renewcommand{\arraystretch}{1.2}
\begin{tabular}{|| c ||c | c  |c | c ||}
\hline
agent &$a, b, c, d, e, f$ & $T$ & $G$ & $M, P$\\
\hline
$\ang$ &$0$ & $\delta_\ang(T)= -\delta_T$ & $\delta_\ang(G)=\delta_G$ & $0$\\
$\dae$ &$0$ & $\delta_\dae(T)=\delta_T$ & $\delta_\dae(G)=-\delta_G$ & $0$\\
\hline
\end{tabular}
}
\]
Consider the uncertainty profile $\mathcal U=\langle\mathcal E,\pertur,  \{T,G\},  \{T,G\}, 1, 1, u\rangle$ where  $u(a,d)= Y(a,d)-k\, r(a,d)$, for some  
$\delta_T>0$ and $\delta_G\ge 0$.  
Letting $\delta= b\delta_T-\delta_G$, the associated zero-sum $\ang/\dae$  game $\Gamma(\mathcal U)$ is given by 

\medskip
{\renewcommand{\arraystretch}{1.2}
\begin{tabular}{cc|c|c|}
\cline{3-4}
& &\multicolumn{2}{|c| }{$\dae$}\\
\cline{3-4}
   & & $\{T\}$    & $\{G\}$\\
\cline{1-4}
\multicolumn{1}{|c|}{\multirow{2}{*} {$\ang$}} &
\multicolumn{1}{|c|}{$\{T\}$}   
& $Y-k\, r$ 
& $(Y-k\, r)+ (1/g)(f-k\, e)\delta$\\
\cline{3-4}
\multicolumn{1}{|c|}{}   &
\multicolumn{1}{|c|}{$\{G\}$}   
& $(Y-k\, r) -(1/g)(f-k\, e)\delta$ 
& $Y-k\, r$\\
\cline{1-4}
\end{tabular}
}

\medskip
When  $(f-k\, e)\delta=0$, the four strategy profiles of $\Gamma(\mathcal U)$ are \PNE, otherwise 
either $(\{T\},\{T\})$ or $(\{G\},\{G\})$ is the unique \PNE of $\Gamma(\mathcal U)$ .
\end{theorem}

\begin{proof}
Let us compute the utilities in $\Gamma(\mathcal U)$. We can easily get 
$Y(\{T\},\{T\}) =Y(\{G\},\{G\}) = Y$ and
$r(\{T\},\{T\}) =r(\{G\},\{G\}) = r$. Recall that the expression of the equilibrium point is,
\begin{eqnarray*}
Y(\{T\}, \{G\}) 
&=& \frac{f}{g}\Big( a+c + G-\delta_G -b(T-\delta_T)\Big) + \frac{d}{g}\frac{M}{P}
= Y +\frac{f}{g}\delta.
\end{eqnarray*}
Thus, summarizing we have
\begin{eqnarray*}
&&Y(\{T\},\{G\}) = Y + (f/g)\delta,\; 
r(\{T\},\{G\}) = r+(e/g)\delta,\\ 
&&
Y(\{G\}, \{T\}) = Y -(f/g)\delta,\;
r(\{G\}, \{T\}) = r-(e/g)\delta.
\end{eqnarray*}
Considering the utilities, clearly 
$u(\{T\},\{T\})=  u(\{G\},\{G\})=Y-k\, r$
and 
\begin{eqnarray*}
u(\{T\},\{G\})
&=& Y(\{T\},\{G\}) - k\, r(\{T\},\{G\})\\ 
&=&Y+(f/g)\delta -k(r+(e/g)\delta)= Y-k\, r+ (1/g)(f-k\, e)\delta.
\end{eqnarray*}
In a similar way  $u(\{G\},\{T\})= (Y-k\, r) - (1/g)(f-k\,e)\delta$ and we get 
the claimed expression for $\Gamma(\mathcal U)$. 
Let us analyze the structure of the \PNE. When $(f-k\,e)\delta=0$ all the pure strategies have the same utility $Y-k\, r$ and therefore all of them are 
\PNE.
When $Y-k\, r>0$ and $(f-k\, e)\delta>0$ the only \PNE is  $(\{T\},\{T\}) $.
 When $Y-k\, r>0$ and $(f-k\, e)\delta<0$ the only \PNE is  $(\{G\},\{G\}) $.
When $Y-k\, r<0$ and $(f-k\, e)\delta>0$ the only \PNE is  $(\{T\},\{T\}) $.
When $Y-k\, r<0$ and $(f-k\, e)\delta<0$ the only \PNE is  $(\{G\},\{G\}) $.



When  $k=Y/r$,  
reconsidering Example \ref{FirstGame} taking $u(a,d)=Y(a,d)-(Y/r)r(a,d)$ and Theorem \ref{k-opposite}, we have that 
$\Gamma(\mathcal U)$ is

\bigskip

{\renewcommand{\arraystretch}{1.2}
\begin{tabular}{cc|c|c|}
\cline{3-4}
& &\multicolumn{2}{|c| }{$\dae$}\\
\cline{3-4}
   & & $\{T\}$    & $\{G\}$\\
\cline{1-4}
\multicolumn{1}{|c|}{\multirow{2}{*} {$\ang$}} &
\multicolumn{1}{|c|}{$\{T\}$}   
& $0$ 
& $(1/g)(f-(Y/r)e)\delta$\\
\cline{3-4}
\multicolumn{1}{|c|}{}   &
\multicolumn{1}{|c|}{$\{G\}$}   
& $ -(1/g)(f-(Y/r)e)\delta$ 
& $0$\\
\cline{1-4}
\end{tabular}
}

\bigskip
When $u(\{T\},\{G\}) >0$, the game has a \PNE in $(\{T\},\{T\})$.
When  $u(\{T\},\{G\}) <0$ the \PNE is  $(\{G\},\{G\})$.
This concludes the proof.
\end{proof}


\section{Uncertainty in the IS-MP Model}
\label{section-IS-MP}
Now we consider the case where the income $Y$ and the inflation $\pi$  become uncertain 
due to perturbations in  $\{\overline{Y}, \pi^*\}$. 
We analyse the  
case where a benevolent $\ang$ tries to keep the income as high as possible, $u_\ang=Y$.
As in   Example \ref{IS-MP-game}, we consider different  views  of $\dae$ in relation to $r$.
To model the case where $\dae$ tries to maximize the inflation, we take 
$u_{\dae} = \pi$. To model the case where $\pi$ is a dis-utility  $\dae$ tries to minimize $\pi$, we take $u_{\dae} = -\pi$. 
The following theorem consider both cases.
%
%
\begin{theorem}
\label{IS-MP-Thm1}
Let $S$ be a perturbation strength model for the IS-MP model with 
$\delta_\ang(e)=\delta_\dae(e)=0$,  for any $e\in \mathcal E\setminus\{\overline{Y}, \pi^*\}$ and
$\delta_\ang(\overline{Y})>0$ and $\delta_\dae(\pi^*)>0$.
Consider and uncertainty profile  
$\mathcal U = \langle \mathcal E,\pertur, \{\overline{Y}, \pi^*\},\{\overline{Y}, \pi^*\}, 1,1, Y , u_{\dae} \rangle$. When $u_{\dae} = \pi$
the $\ang/\dae$ game is

\medskip
{\renewcommand{\arraystretch}{1.2}
\begin{tabular}{cc|c|c|}
\cline{3-4}
& &\multicolumn{2}{|c| }{$\dae$}\\
\cline{3-4}
\rule{0pt}{11pt}
   & & $\{\overline{Y}\}$   & $\{\pi^*\}$\\
\cline{1-4}
\multicolumn{1}{|c|}{\multirow{2}{*} {$\ang$}} &
\multicolumn{1}{|c|}{\rule{0pt}{11pt}$\{\overline{Y}\}$}   
& $Y+\delta_\ang(\overline{Y})+\delta_\dae(\overline{Y}),\;\pi$
& $Y+\delta_\ang(\overline{Y}),\;  \pi +\delta_\dae(\pi^*)$\\
\cline{3-4}
\multicolumn{1}{|c|}{}   &
%
\multicolumn{1}{|c|}{\rule{0pt}{11pt}$\{\pi^*\}$}   
& $ Y+\delta_\dae(\overline{Y}), \;\pi +\delta_\ang(\pi^*)$
& $Y,\; \pi+\delta_\ang(\pi^*) +\delta_\dae(\pi^*)$\\
\cline{1-4}
\end{tabular}
}

\medskip
$\{\overline{Y}\}$ is the  dominant strategy for $\ang$ and $\{\pi^*\}$ is the  dominant strategy for $\dae$.
The unique \PNE  is $(\{\overline{Y}\},\{\pi^*\})$ with
$Y(\{\overline{Y}\},\{\pi^*\})= Y+\delta_\dae(\overline{Y})$ and $\pi(\{\overline{Y}\},\{\pi^*\})=\pi +\delta_\ang(\pi^*)$.
When  $u_{\dae} = -\pi$  the game is 

\medskip
{\renewcommand{\arraystretch}{1.2}
\begin{tabular}{cc|c|c|}
\cline{3-4}
& &\multicolumn{2}{|c| }{$\dae$}\\
\cline{3-4}
\rule{0pt}{11pt}
   & & $\{\overline{Y}\}$   & $\{\pi^*\}$\\
\cline{1-4}
\multicolumn{1}{|c|}{\multirow{2}{*} {$\ang$}} &
\multicolumn{1}{|c|}{\rule{0pt}{11pt}$\{\overline{Y}\}$}   
& $Y+\delta_\ang(\overline{Y})+\delta_\dae(\overline{Y}),\;-\pi$
& $Y+\delta_\ang(\overline{Y}),\;  -\pi +\delta_\dae(\pi^*)$\\
\cline{3-4}
\multicolumn{1}{|c|}{}   &
%
\multicolumn{1}{|c|}{\rule{0pt}{11pt}$\{\pi^*\}$}   
& $ Y+\delta_\dae(\overline{Y}), \;-\pi +\delta_\ang(\pi^*)$
& $Y,\; -\pi+\delta_\ang(\pi^*) +\delta_\dae(\pi^*)$\\
\cline{1-4}
\end{tabular}
}

\medskip
The  dominant strategy for both $\ang$ and $\dae$ is $\{\overline Y\}$ and
the only \PNE is $(\{\overline{Y}\},\{\overline{Y}\})$. In this case 
$Y(\{\overline{Y}\},\{\overline{Y}\}) =Y+\delta_\ang(\overline{Y})+\delta_\dae(\overline{Y})$ and  
$\pi(\{\overline{Y}\},\{\overline{Y}\})=\pi$. 
\end{theorem}

\begin{proof}
Let us consider the game $\Gamma(\mathcal U)$. 
According to Lemma \ref{JustPerturbationsIS-MP} we have
\[u_\ang(a,d)=Y(a,d) = Y+ \delta_{\pertur}(\overline{Y})[a,d],\;
u_\dae(a,d)=\pi(a,d) = \pi+\delta_{\pertur}(\pi^*)[a,d]
\] 
The set of strategy profiles is:
\[
A_\ang\times A_\dae=\{(\{\overline{Y}\},\{\overline{Y}\}),
(\{\overline{Y}\},\{\pi^*\}),(\{\pi^*\},\{\overline{Y}\}),(\{\pi^*\},\{\pi^*\})\}
\]
This gives the following bi-matrix form.

\bigskip
{\renewcommand{\arraystretch}{1.2}
\begin{tabular}{cc|c|c|}
\cline{3-4}
& &\multicolumn{2}{|c| }{$\dae$}\\
\cline{3-4}
\rule{0pt}{11pt}
   & & $\{\overline{Y}\}$   & $\{\pi^*\}$\\
\cline{1-4}
\multicolumn{1}{|c|}{\multirow{2}{*} {$\ang$}} &
\multicolumn{1}{|c|}{\rule{0pt}{11pt}$\{\overline{Y}\}$}   
& 
$Y+\delta_\pertur(\overline{Y})[\{\overline{Y}\},\{\overline{Y}\}], \;
\pi + \delta_\pertur(\pi^*)[\{\overline{Y}\},\{\overline{Y}\}]$
& $\dots,\dots$\\
\cline{3-4}
\multicolumn{1}{|c|}{}   &
\multicolumn{1}{|c|}{\rule{0pt}{11pt}$\{\pi^*\}$}   
& $ Y+\delta_\pertur(\overline{Y})[\{\pi^*\},\{\overline{Y}\}], \;
\pi + \delta_\pertur(\pi^*)[\{\pi^*\},\{\overline{Y}\}] $
& $\dots,\dots$\\
\cline{1-4}
\end{tabular}
}

\bigskip
Consider the strategy profile $(a,d)=(\{\overline{Y}\},\{\overline{Y}\})$. 
According to the Definition \ref{DefDeltaStress}, as $\overline{Y}\in a\cap d=\{\overline{Y}\}$, it holds  $ \delta_{\pertur}(\overline{Y})[a,d]= \delta_\ang(\overline{Y})+\delta_\dae(\overline{Y})$ and 
$u_\ang(\{\overline{Y}\},\{\overline{Y}\})= Y + \delta_\ang(\overline{Y})+\delta_\dae(\overline{Y})$.
As $\pi\not\in a\cup d=\{\overline{Y}\}$, the perturbation strength verifies $\delta_{\pertur}(\pi^*)[a,d]=0$ and therefore 
$u_\dae(\{\overline{Y}\},\{\overline{Y}\})= \pi$. The remaining  cases are similar and thus we get the claimed bi-matrix form for 
$\Gamma(\mathcal U)$.
Let us analyse  dominance. For $\ang$, as $\delta_\ang(\overline{Y})>0$, independently of the sign of 
$\delta_\dae(\pi^*)$, best responses verify
$B_\ang(\overline{Y})=B_\ang(\pi^*) = \overline{Y}$.
For $\dae$, as $\delta_\dae(\pi^*)>0$, best responses are
$B_\dae(\overline{Y})=B_\dae(\pi^*) = \pi^*$. The proof for $\mathcal U_2$ is similar and the claim  follows. 
\end{proof}

\bigskip
It is worth observing  that for $\langle \mathcal E,\pertur, \{\overline{Y}, \pi^*\},\{\overline{Y}, \pi^*\}, 1,1,Y, \pi\rangle$, the   
selfish behaviour of $\ang$ and $\dae$ give a \PNE $(\{\overline{Y}\},\{\pi^*\})$ corresponding to a "moderate" perturbed  values i.e., $Y(\{\overline{Y}\},\{\pi^*\})=Y+\delta_\ang(\overline{Y})$ and 
$\pi(\{\overline{Y}\},\{\pi^*\})=\pi+\delta_\dae(\pi^*)$. 
Selfish behaviour protect us from extreme behaviours with maximal income or maximal inflation. The two extreme cases have utilities  $\big(Y+\delta_\ang(\overline{Y})+\delta_\dae(\overline{Y}),\pi \big)$ and $\big(Y, \pi+\delta_\ang(\pi^*)+\delta_\dae(\pi^*)\big)$.
Situation changes when
$\langle \mathcal E,\pertur, \{\overline{Y}, \pi^*\},\{\overline{Y}, \pi^*\}, 1,1,Y, -\pi\rangle$.  
In this case situation evolves to  a \PNE $(\{\overline{Y}\},\{\overline{Y}\})$ such that 
$Y(\{\overline{Y}\},\{\overline{Y}\}) =Y+\delta_\ang(\overline{Y})+\delta_\dae(\overline{Y})$ and  
$\pi(\{\overline{Y}\},\{\overline{Y}\})=\pi$. In such a case $\ang$ and $\dae$ move on to get high income $Y$ and to do not touch at the  inflation $\pi$.

A possible practical implication of the Theorem \ref{section-IS-MP} follows. The perturbation model $\pertur$ acts over  $\overline{Y}$ and $\pi^*$ through
$\delta_\ang(\overline{Y})>0$, $\delta_\ang(\pi^*) $,  $\delta_\dae(\overline{Y})$, $\delta_\dae(\pi^*)>0$. When  $u_{\dae} = \pi$ the \PNE is $(\{\overline Y\},\{\pi^*\})$,
when  $u_{\dae} = -\pi$ the \PNE is $(\{\overline Y\},\{\overline Y\})$ it  is worth noting that the Nash equilibria are independent of the signs of $\delta_\ang(\pi^*)$ and  $\delta_\dae(\overline{Y})$. As a consequence any policy trying to act over $\delta_\ang(\pi^*)$ or  $\delta_\dae(\overline{Y})$ will not modify the general structure of the Nash equilibria.

We conclude our  study considering a subclass of zero-sum \angdae games generalizing  the conditions of the uncertainty profiles considered in Example \ref{IS-MP-game}.

\begin{theorem}
\label{IS-MP-ThmZeroSum}
Let $S$ be a perturbation strength model for the IS-MP model such that $\delta_\ang(e)=\delta_\dae(e)=0$,  for any $e\in \mathcal P$.
Let $\mathcal U = \langle \mathcal E,\pertur, \mathcal A,\mathcal D, 1,1, u\rangle$ where
\[u(a,d)= Y(a,d)-(Y(\{\emptyset\},\{\emptyset\})/\pi(\{\emptyset\},\{\emptyset\}))\pi(a,d).\] 
The associated zero-sum $\ang/\dae$ game $\Gamma(\mathcal U)$ is:

\medskip
{\renewcommand{\arraystretch}{1.2}
\begin{tabular}{cc|c|c|}
\cline{3-4}
& &\multicolumn{2}{|c| }{$\dae$}\\
\cline{3-4}
   & & \rule{0pt}{11pt}$\{\overline{Y}\}$   & $\{\pi^*\}$\\
\cline{1-4}
\multicolumn{1}{|c|}{\multirow{2}{*} {$\ang$}} &
\multicolumn{1}{|c|}{\rule{0pt}{11pt}$\{\overline{Y}\}$}   
& 
$\delta_\ang(\overline{Y})+\delta_\dae(\overline{Y})$
& $\delta_\ang(\overline{Y})-(Y/M)\delta_\dae(\pi^*)$\\
\cline{3-4}
\multicolumn{1}{|c|}{}   &
\multicolumn{1}{|c|}{\rule{0pt}{11pt}$\{\pi^*\}$}   
& $ \delta_\dae(\overline{Y})-(Y/\pi)\delta_\ang(\pi^*)$
& $-(Y/\pi)(\delta_\ang(\pi^*) +\delta_\dae(\pi^*))$\\
\cline{1-4}
\end{tabular}
}

\bigskip
When $\delta_\agent(e)>0$, for $\agent\in\{\ang,\dae\}$ and $e\in \{\overline{Y},\pi^*\}$,
$\{\overline{Y}\}$ is the  dominant strategy for $\ang$ and $\{\pi^*\}$ is the  dominant strategy for $\dae$.
Furthermore, 
$(\{\overline{Y}\},\{\pi^*\})$ is the unique \PNE.
\end{theorem}
  
\begin{proof} The tabular form of $\Gamma(\mathcal U)$ comes straightforward from 
the tabular form appearing in the Theorem \ref{IS-MP-Thm1}. With respect to  dominance, let us compute  best responses. 
The $\ang$ best response to  $\{\overline{Y}\}$, as  $\delta_\ang(\overline{Y})+\delta_\dae(\overline{Y})> \delta_\dae(\overline{Y})-(Y/\pi)\delta_\ang(\pi^*)$,  verifies 
$B_\ang(\{\overline{Y}\})= \{\overline{Y}\}$ and similarly $B_\ang(\{\pi^*\})= \{\overline{Y}\}$. Consider now the  best responses of $\dae$. When $\ang$ chooses 
$\{\overline{Y}\}$, as $\delta_\ang(\overline{Y})+\delta_\dae(\overline{Y})>
\delta_\ang(\overline{Y})-(Y/M)\delta_\dae(\pi^*)$ and $\dae$ tries to minimize the utility, we get $B_\dae(\{\overline{Y}\})= \{\pi^*\}$. Similarly, when $\ang$ chooses $\{\pi^*\}$ we get 
 $B_\dae(\{\pi^*\})= \{\pi^*\}$.
\end{proof}


\section{Conclusions and Further Developments}
\label{Conclusions}
We have shown how to adapt  the $\ang/\dae$-framework provided in \citep{GSS2014-CJ} to analyse  uncertainty in the IS-LM and the IS-MP models. In both cases, we have studied different possible cases of uncertainty 
through the set of Nash equilibria showing the applicability of the  framework. 
In general, a common way to deal with uncertainty consists on {\em modelling several scenarios} and develop  qualitative likelihood analysis of the different cases.
In our proposed analysis positive and negative aspects are taken into account.

Our approach can be adapted to analyse uncertainty in other financial settings. As an example, consider an option valuation using the Greeks \citep{HullOptions}. We develop now a possible application based on  $\Delta$ where
$\Delta$ is the ratio between the changes in the price of the derivative to the corresponding price of the underlying asset.
We do that considering uncertainty in the price $f$  of a call option over a stock $S$. We use the   
one step binomial tree model \citep{HullOptions} having the following  components:
\[
{\renewcommand{\arraystretch}{1.2}
\begin{tabular}{|| c |c |c | c|c|c||}
\hline
 Stock price & $S$ & Time period & $T$ &Up jump & $u$\\
Strike price & $X$ & Risk-free rate & $r$ & Down jump &$d$\\
\hline
\end{tabular}
}
\]
Where $S$ is the price of the stock,  $X$ is the strike price (the price at the expiration date), 
$T$ is the time period before expiration. The parameters $u$ and $d$ fixed such that, if stock goes up the price (of the stock) should be $S\times u$ and if it goes down the price should be $S\times d$.  
Consider the following valuation taken from \citep{HullOptions} serving us to introduce the basic constructs:
\[
{\renewcommand{\arraystretch}{1.2}
\begin{tabular}{||c  c  c || c   c  c ||}
\hline
$S$ & $X$ & $T$ & $r$ & $u$& $d$\\
\hline
$20$ & $21$ &  $0.25$ &$0.12$  & $1.1$& $0.6$ \\
\hline
\end{tabular}
}
\]
In order to compute $f$, consider a portfolio having a long $\Delta$ shares and one short call option.
If the stock goes up in 3 months ($T= 0.25$)  it holds $S\times u>X$ and  the payoff from the derivative is  $f_u=S\times u -X=1$. The value of the portfolio in the up case is  $S\times u\times \Delta -f_u $.
If the stock goes down, $f_d=0$ the value is  $S\times d\times\Delta-f_d$.  
In order to be risk-free $\Delta$ needs to verify
$S\times u\times \Delta -f_u =  S\times d\times \Delta-f_d$. Therefore 
\[\Delta = \frac{f_u-f_d}{S\times u-S\times d}=\frac{1}{10}\]
and the value of the portfolio at the end of the period will be $S\times d\times \Delta$.
To get the present value $v$ we discount at free-risk rate 
$v=S\times d\times \Delta \times e^{-r\times T}=1.164535$
and  the price $f$ of the call is 
$f=  S\times\Delta \times (1-v)=S\times\Delta \times (1-d\times e^{-r\times T})=0.835465$.
As in a call the values $S$, $X$ and $T$ are  explicitly is written, therefore is no ambiguity about them. Therefore the set of exogenous components is  $\mathcal E= \{ r, u,d\}$. An example of a  perturbation strength model $\pertur$ is 
\[
{\renewcommand{\arraystretch}{1.2}
\begin{tabular}{|| c || c | c | c ||}
\hline
$\agent$ & $r$ & $u$& $d$\\
\hline
\hline
$\ang$  &$-0.05$ & $+0.4$ & $0$   \\
$\dae$ & $+0.10$ & $0$    & $+0.3$ \\
\hline
\end{tabular}
}
\]
Given the strategy profile $(a,d)$ the new values of the perturbed set are
$\mathcal E'= \{ r', u',d'\}$. Let, consider the case $(a,d)=(\{u\},\{d\})$.
As $r\not\in \{u,d\}$ it holds is $r'=r$.
As $u$ is only selected by the angel, $u'=u+\delta_\ang(u)=1.1+0.4=1.5$.
As $d$ is just selected by the daemon $d'=d+\delta_\dae(d)=0.6+0.3=0.9$.
As $S\times u'=   30$ and $f_{u'}= 30-21=9$. As $S\times d'= 18$ we have $f_{d'}=0$.
As usual we shorten $\Delta(\stress_\pertur[\{u\},\{d\}])$ as $\Delta(\{u\},\{d\})=3/4$ 
%
%
and 
$f(\{u\},\{d\})=S\times\Delta(\{u\},\{d\}) \times (1-d'\times e^{-r'\times T})=1.898985$.
As another case example of interest consider the strategy where both $\ang$ and $\dae$ choose to strength $r$. In this case
that is $(a,d)=(\{r\},\{r\})$. In this case $u'=u=1.1$, $d'=d=00.6$ and $r'=r+\delta_\ang(r)+\delta_\dae(r)=0.12-0.05+0.10=0.17$. As neither $u$ nor $d$ has been perturbed,
$\Delta(\{r\},\{r\})=\Delta = 1/10$ and $f(\{r\},\{r\})=0.8499314$.
Let 
$\mathcal U=\langle \mathcal E,\pertur, \{\{u\} ,\{r\}\}, \{\{d\},\{r\}\}, 1,1,f\rangle$
be an uncertainty profile for a call option. The corresponding $\ang/\dae$ game is:

\smallskip
{\renewcommand{\arraystretch}{1.2}
\begin{tabular}{cc|c|c|}
\cline{3-4}
& &\multicolumn{2}{|c| }{$\dae$}\\
\cline{3-4}
   & & $\{d\}$    & $\{r\}$\\
\cline{1-4}
\multicolumn{1}{|c|}{\multirow{2}{*} {$\ang$}} &
\multicolumn{1}{|c|}{$\{u\}$}   
& $1.898985$ 
& $4.321089$\\
\cline{3-4}
\multicolumn{1}{|c|}{}   &
\multicolumn{1}{|c|}{$\{r\}$}   
& $ 0.5780649$ 
& $0.8499314$\\
\cline{1-4}
\end{tabular}
}

\medskip
\noindent
The only \PNE is $(\{u\},\{d\})$ with $\nu(\mathcal U)=1.898985$.
Therefore, in the uncertain situation described by $\mathcal U$ perhaps the price of the call should be $1.898985$ rather than $0.835465$.

Studying the applicability of the \angdae approach in other macroeconomic settings is of great interest. In particular we are interested in compare our approach to other provided analytical tools. To do so we would need to further analyse  the structure and properties of the Nash equilibria of the associated games.
It is important to  know the possibilities and limits of the uncertainty profiles and $\ang/\dae$ games in macroeconomics. 
As  pointed in \citep{Durlauf2012}, a way to decide with no probabilities is to guard against really bad cases. In such a setting regret analysis  provides a way to analyse such cases. 
Analysing the connection between \angdae and regret analysis   would be interesting. 
In \citep{TheClimateCasino}, different scenarios are developed  in order to analyse different possibilities in relation to climate change. For instance, when considering the cost of meeting global temperature targets, two possible scenarios for  the average-cost/temperature are considered. Those models are not far away from the ones considered in this paper,  thus uncertainty profiles and $\ang/\dae$ could provide a complementary analysis tool. 
%
\citet{Var77} has studied the stability of the IS-LM model. As
$(Y(a,d), r(a,d))$ can be seen as perturbations of the equilibrium point $(Y,r)$ (see Lemmas \ref{linearity-IS-LM}, \ref{EasyFiscalPolicies})
if will be of interest to study the relationship among of both models. In general, it would be interesting to try to apply $\ang/\dae$-framework to non-linear models. Ever if general theorems seems difficult to obtain, accurate numerical examples could be a first step to perform the comparison.  Moreover, in order to transform $\ang/\dae$-framework into a practical tool  it seems necessary to explore and develop the links with macroeconomic policy coordination and international macroeconomic integration.

It is also important to  know the possibilities, limits and weaknesses of the $\ang/\dae$ approach. 
We have considered uncertainty profiles
$\mathcal U=\langle \mathcal E,\pertur,\mathcal A,\mathcal D, 
 b_{\ang}, b_{\dae},  u_{\ang}, u_{\dae}\rangle$
where the spreads $ b_{\ang}$,  $b_{\dae}$ are fixed independently of the other parameters in $\mathcal U$. It is interesting to consider the case where the spreads are related to other parameters (see the first part of the Lemma \ref{LemmaNoFreedom}), for instance  $b_{\ang} = \lfloor \frac{1}{2}\#\mathcal A\rfloor$ or 
 $b_{\dae} = \lfloor \frac{2}{3}\#\mathcal A\rfloor$.
Moreover, any $\mathcal U$ is a tuple
having many parameters to be determined. In principle those parameters are fixed by the analyser based on his perception of the world an historical data. However the danger of  over parametrization exists \citep{HullRisk}. Remind the von Neumann dictum, ``with  four parameters I can fit an elephant and with five I can make him wiggle his trunk'' \citep{DysonElephant}. 
Also possible connections  between uncertainty profiles and  $\ang/\dae$ games and the maxmin expected utility approach proposed by 
\citep{GilboaSchmeidler1989} and developed in \citep{Hansen2006, HansenSargentAlter2006} deserves further study. 

Finally, in this paper we have consider only short-time models. The $\ang/\dae$ approach
for long-time models is left open. A common way to deal with uncertainty along the time consists on {\em modelling the temporal evolution} of several possible scenarios.
As an example, consider the two possible scenarios  in  the
evolution (France, 1820-2100) of the annual value of bequest and gifts \citep{PikettyCapital}.
Scenarios are given in terms of the growth of output $g$ and of the rate of return on capital $r$ (do not confuse with the previous interest rate $r$).
There is a {\em central scenario} with $g= 1.7\%$ and $r=3.0\%$ and an {\em alternative scenario with}   $g= 1.0\%$ and $r=5.0\%$.
In order to tackle such cases,  temporal aspects can be incorporated to  game theory using stochastic games \citep{ShapleyStochastic}. 
In \citep{CastroGSS15} a  time dimension of the $\ang/\dae$ approach was integrated into the frame of stochastic automata. Perhaps this approach, based on stochastic automata,  could be adapted to deal with the uncertainty along the time of some economic models.




\section*{Acknowledgements}
We would like to thank the anonymous referees for their  careful reading and for their fruitful suggestions that help us to improve the paper. 
J. Gabarro and M. Serna are
partially supported by
funds from the Spanish Ministry for Economy and Competitiveness (MINECO) and the European Union (FEDER funds) under grant TIN2013-46181-C2-1-R (COMMAS)
and from AGAUR, Generalitat de Catalunya under grant SGR 2014:1034 (ALBCOM).



%

%

\end{document}